\documentclass[11pt]{article}
\usepackage{graphicx}
\usepackage{amssymb}
\usepackage{amsmath}
\usepackage{amsthm}
\usepackage{multicol,multirow}
\usepackage{url}
\usepackage{hyperref}
\usepackage{color}
\usepackage{lscape}
\usepackage[margin=3.5cm]{geometry}
\usepackage{setspace}

\hypersetup{colorlinks,linkcolor={blue},citecolor={blue},urlcolor={red}}

\newtheorem{theorem}{Theorem}

\begin{document}

\title{Empirical Likelihood Inference for Sen and Sen-Shorrocks-Thon Indices}

\author{
Sreelakshmi N\\
\small Department of Statistics, Prajyotiniketan College, Kerala, India
\and
Saparya Suresh\\
\small Decision Sciences \& Operations Management, Indian Institute of Management, Kozhikode, India
\and
Sudheesh K. Kattumannil\\
\small Applied Statistics Unit, Indian Statistical Institute, Chennai, India
}

\date{}
\maketitle

\begin{abstract}
The Sen index and Sen-Shorrocks-Thon (SST) index are widely used measures of poverty indices. Developing reliable inference for these measures enables us to compare these measures in different populations of interest in an effective way. It is important to construct confidence intervals for the Sen index and SST index, which provide better coverage probability and shorter interval length.

Motivated by this, we discuss empirical likelihood (EL) and jackknife empirical likelihood (JEL) based inference for the Sen index. To derive a JEL-based confidence interval for the Sen and SST indices, we propose a new estimator for the Sen index using the theory of U-statistics and examine its properties.

The large sample properties of the EL and JEL ratio statistics are studied. We also discuss EL and JEL-based inference for the Sen-Shorrocks-Thon (SST) index. The finite sample performance of the EL and JEL-based confidence intervals of both Sen and SST indices is evaluated through a Monte Carlo simulation study.

Finally, we illustrate our methods using individual-level data from the Panel Study of Income Dynamics (PSID) survey from the US as well as Indian household-level income data for different states sourced from the Consumer Pyramids Household Survey (CPHS).
\end{abstract}

\noindent \textbf{Keywords:} Empirical likelihood; JEL; Sen index; Sen-Shorrocks-Thon index.
\section{Introduction}\label{sec1}
The measurement of poverty and inequality evaluates whether individuals in an economy lack the resources necessary to meet basic needs and quantifies disparities in welfare across the population. Accurate poverty measurement is essential for empirical welfare analysis and for assessing progress toward international development objectives. Over time, poverty measures have evolved from simple descriptive indicators to axiomatically grounded and distribution-sensitive indices.

Early contributions focused on poverty as the shortfall of income relative to a predetermined poverty line, leading to simple measures such as the headcount ratio and the poverty gap index. The headcount ratio captures the incidence of poverty, whereas the poverty gap index measures the average shortfall of poor individuals from the poverty line. Although intuitive and easy to compute, these measures neglect inequality among the poor and fail to satisfy desirable transfer-sensitivity properties.

A major advance was made by Watts (1968), who introduced a continuous poverty index based on the logarithm of income and interpreted poverty as a social welfare loss. This contribution initiated the axiomatic approach to poverty measurement, which was formalized by Sen (1976). Sen proposed that a satisfactory poverty measure should satisfy properties such as focus, monotonicity, and transfer sensitivity, and consequently introduced the Sen index, which simultaneously captures incidence, intensity, and inequality within a unified framework.

Subsequent work extended and refined this approach. Kakwani (1980) proposed a general class of poverty measures incorporating distributional sensitivity, while Clark, Hemming, and Ulph (1981) developed social welfare--based index classes emphasizing ethical consistency and decomposability. Foster, Greer, and Thorbecke (1984) introduced a decomposable class of poverty measures that became particularly influential due to its tractability and subgroup decomposition properties. Parallel efforts aimed to reconcile Sen’s conceptual framework with decomposability. Thon (1979, 1983) proposed early modifications, and Shorrocks (1995) later reformulated the Sen index into what is now known as the Sen--Shorrocks--Thon (SST) index. The SST index preserves Sen’s tripartite structure---incidence, intensity, and inequality---while achieving a multiplicative and decomposable form.

Among the large class of poverty measures, the Sen and SST indices remain particularly influential because of their axiomatic foundations and distributional sensitivity. Both indices combine the headcount ratio, the mean income gap, and a Gini-type component among the poor. Despite their conceptual appeal, their statistical properties are nontrivial. In particular, their sampling distributions are analytically complex, especially under finite samples or complex survey designs. The conventional approach relies on plug-in estimators combined with first-order asymptotic approximations for variance estimation. However, plug-in procedures may exhibit non-negligible bias in small or moderate samples, leading to unreliable inference.

Despite their extensive use in empirical welfare analysis, the Sen and SST indices lack a transparent representation within the classical U-statistic framework. Existing inference procedures obscure the structural source of variability and complicate higher-order refinement. In particular, the interaction between the headcount component, the income gap term, and the Gini-type inequality measure among the poor induces a nontrivial dependence structure that is not immediately amenable to standard variance derivations.

The present paper addresses this gap by establishing an explicit symmetric kernel representation of the Sen and SST indices, thereby embedding these measures within the unified theory of U-statistics. This formulation permits a direct application of Hoeffding decomposition, yields a tractable influence function, and produces a closed-form asymptotic variance expression under mild regularity conditions. Building on this structure, we further develop empirical likelihood (EL) and jackknife empirical likelihood (JEL) inference procedures and establish Wilks-type limiting distributions without requiring explicit variance estimation. To the best of our knowledge, such a unified U-statistic representation and likelihood-based inference framework for the Sen and SST indices have not previously been developed. The resulting methodology provides both theoretical clarification and practical inferential advantages for distribution-sensitive poverty measures.

The remainder of the paper is organized as follows. Sections 2 and 3 introduce the proposed U-statistic estimators for the Sen and SST indices, respectively, and derive their asymptotic properties. Section 4 presents a Monte Carlo study comparing the finite-sample performance of the proposed estimators with that of Davidson (2009). Section 5 provides empirical illustrations using data from the Panel Study of Income Dynamics (PSID) and the Consumer Pyramids Household Survey (CPHS). Section 6 concludes.
\section{Inference for the Sen index}
We obtain a new non-parametric estimator of Sen Index.  To understanding the formulation of Sen index, we need to understand three components of it - the headcount ratio, the income gap ratio and the Gini index. We  first review each component before defining the Sen index.

The headcount ratio $(H)$ is a measure of the proportion of people falling below the poverty line $z$. It is given by
\begin{equation*}
H = \frac{q}{N},
\end{equation*}
where $q$ represents the number of individuals living below the poverty line, $z$, and $N$ represents population size. While the incidence of poverty quantifies the number of individuals below the poverty line, it does not adequately measure the intensity of the issue.

Thus, we have the second kind of poverty measure, which looks at the distribution of impoverished individuals, and a popular measure in this category is the income gap ratio $(I)$, defined as
\begin{equation*}
I=1-\frac{\mu_z}{z},
\end{equation*}
where $\mu_z$ represents the mean income of the poor. Note that the income gap ratio captures the intensity of poverty in a population.

Takayama (1979) emphasized the significance of incorporating income inequality measures in poverty studies  by introducing the Gini index for the poor as a poverty metric. If $X$ is a continuous random variable with distribution function $F(.)$, the Gini index for the poor is defined as

\begin{equation*}\label{gini}
G_p=2   \mu_z^{-1}\int_{0}^{z} dzyF(y)dF(y) - 1.
\end{equation*}

From a policy-making perspective, it would be beneficial to establish an index that encompasses the three dimensions of incidence, intensity, and inequality.  Sen (1976) was the first to propose a poverty index, referred to as the Sen index, which is based on the measures of incidence, intensity, and inequality among the poor. The Sen index is defined as
\begin{equation}\label{sendef}
S=HI+\frac{q}{q+1}(1-I)G_p.
\end{equation}
The above equation \eqref{sendef} be rewritten as (Davidson (2009))
\begin{eqnarray}\label{senfinal}
\widehat{S}_1&=&\frac{2}{zF(z)}\int_0^{z}(z-x)(F(z)-F(x)dF(x).
\end{eqnarray}
Based on the above definition, we propose a new estimator of the Sen index in this section. We also develop an EL and JEL-based confidence interval for Sen index.
Let $F$ be the cumulative distribution function of the income under study. Suppose  $X_1,\ldots,X_n$ is random sample of size $n$ drawn from $F$.
A plug-in estimator of $S$ defined in equation \eqref{senfinal} is given by
\begin{eqnarray}
\label{senest}
\widehat{S}_1&=&\frac{2}{zF_n(z)}\int_0^{z}(z-x)(F_n(z)-F_n(x)dF_n(x)\nonumber\\&=&\frac{2}{nqz}\sum_{i=1}^{q}(z-X_{(i)})\left(q-i\right),
\end{eqnarray}
where $X_{(i)},\, i=1,\ldots,n$ are the $i$-th order statistics based on a random sample $X_1,\ldots,X_n$ from $F$, $q$ represents the number of individuals living below the poverty line $z$  and $F_n(.)$ represents the empirical distribution function.

Bishop et al. (1997) derived U-statistic estimators for each component (headcount, mean gap, and a Gini component) of the Sen index and studied large sample properties of the estimator. However, these estimators were complex and difficult to implement.
Davidson (2009) proposed  an estimator for $S$ which satisfy the population symmetry axiom and is given by
\begin{equation*}\label{sendav}
\widehat{S}_2=\frac{2}{nqz}\sum_{i=1}^{q}(z-X_{(i)})\left(q-i+\frac{1}{2}\right).
\end{equation*}
The estimator $\widehat{S}_2$ of $S$ is derived by duplicating the population, resulting in income $X_{(i)}$ appearing twice with rankings $2i-1$ and $2i$ accordingly. Davidson (2009) noted that deriving an expression for bias and standard error of the estimator for further inference is challenging.
To fill the gap in the literature, we present a simple estimator for the Sen index based on the theory of U-statistics and examine its properties, which are elaborated in the subsequent subsection. Additionally, we detail the constructions of EL and JEL-based confidence intervals for the Sen index.

\subsection{Non-parametric estimator of Sen index}

  Let $X_1$ and $X_2$ be independently distributed random variables from $F$. Next, we obtain a simple estimator for the Sen index. For this purpose, consider
\begin{eqnarray}\label{senrewrite}
E\big((z-X_1)I(X_1<X_2 \le z)\big)&=& \int_{0}^{\infty}\int_{0}^{\infty}(z-x)I(x<y<z)dF(y)dF(x)\nonumber\\
&=&\int_{0}^{z}(z-x)\big(\int_{x}^{z}dF(y)\big)dF(x)\nonumber\\
&=&\int_{0}^{z}(z-x)\big(F(z)-F(x)\big)dF(x).
\end{eqnarray}
Using equation (\ref{senrewrite}), the Sen index given in equation  (\ref{senfinal}) can be rewritten as
\begin{equation}\label{21}
S=\frac{2}{z}\frac{\Delta_1}{\Delta_2},
    \end{equation}where
\begin{equation*}
  \Delta_1=E\big((z-X_1)I(X_1<X_2 \le z)\big)
\end{equation*}
and
\begin{equation*}
  \Delta_2=P(X_1 \le z).
\end{equation*}
Given a random sample $X_1,X_2,\ldots, X_n$ of size $n$ from $F$. An estimator of $\Delta_1$ based on U-statistics is given by
\begin{equation*}
  U_{1}=\frac{1}{\binom{n}{2}}\sum_{i=1}^{n}\sum_{j=1,j<i}^{n}\psi_{1}(X_i,X_j),
\end{equation*}
where $\psi_{1}(X_1,X_2)$ is a symmetric kernel defined as
$$\psi_{1}(X_1,X_2)=\frac{(z-X_1)I(X_1<X_2 \le z)+(z-X_{2})I(X_2<X_1 \le z)}{2}.$$
We write the estimator of $\Delta_2$  as a U-statistic estimator with a kernel of degree 2 as
\begin{equation*}
  U_{2}=\frac{1}{\binom{n}{2}}\sum_{i=1}^{n}\sum_{j=1,j<i}^{n}\psi_{2}(X_i,X_j),
\end{equation*}
where
 $$\psi_{2}(X_1,X_2)=\frac{I(X_1 \le z)+I(X_2 \le z)}{2}.$$
 Thus, the estimator for $S$ is given by
 \begin{equation}\label{senu}
  \widehat{S}=\frac{2}{z}\frac{U_1}{U_2}.
 \end{equation}Next, we simplify the expression in (\ref{senu}) using order statistics. Note that $q$ is the number of  people below the poverty line $z$. Let $X_{(i)},\,i=1,\ldots,n$ be the $i$-th order statistics based on a random sample $X_1,X_2,\ldots, X_n$ of size $n$ from $F$.      After simplification, from  (\ref{senu}), we obtain 
 \textit{\begin{equation}\label{senustat}
 \widehat{S}=\frac{1}{(n-1)qz}\sum_{i=1}^{q}\big(z(q-1)-2(q-i)X_{(i)}\big).
 \end{equation}}

\par  Next, we study the asymptotic properties of the estimator $ \widehat{S}$.  In the following theorem, we prove the consistency of $\widehat{S}$.
 \begin{theorem}
   Given $z$, as $n\rightarrow \infty$, $\widehat{S}$ converges in probability to $S$.
 \end{theorem}
\begin{proof}
    Since $U_1$ is a U-statistic, as $n\rightarrow\infty$, $U_1$ converges in probability to $\Delta_1$ (Lehmann, 1951) and by the law of large numbers, $U_2$ converges in probability to $\Delta_2.$  Since we can write
\begin{equation}
 \frac{\widehat{S}}{S}=\frac{U_2}{\Delta_2}\frac{\Delta_1}{U_1},
 \end{equation}the proof follows.
\end{proof}
\noindent The proofs of the following theorems are given in the Appendix.
\begin{theorem}Define $Q(u)=\inf\{x:F(x)\ge u\}$ and $p=F(z)$.  Let $Q(u)$ be continuous on $[0,p]$ and $ \int_0^{p} |Q(u)| < \infty$. For a given $z$, $\widehat{S}$ is an asymptotically unbiased estimator of the Sen index $S$.
\end{theorem}

\begin{theorem}\label{thm_sen2}
Given $z$, as $n\to\infty$, $\sqrt{n}(\widehat S-S)$ converges in distribution to a Gaussian random variable with mean zero and variance $\sigma^2$, where
$$
\sigma^2
=
\frac{4}{z^2}
\left(
\frac{1}{F^2(z)}\sigma_1^2
+
\frac{\Delta_1^2}{F^4(z)}\sigma_2^2
-
\frac{2\Delta_1}{F^3(z)}\sigma_{12}
\right)
$$
with
\begin{align*}
\sigma_1^2
&=
Var\!\left(
X(F(X)-F(z))-\int_0^X y\,dF(y)
\right),\\[0.5em]
\sigma_2^2
&=
\frac14 F(z)\{1-F(z)\},\\[0.5em]
\sigma_{12}
&=
\frac12
\int_0^z
\left[
x\{F(x)-F(z)\}
-
\int_0^x y\,dF(y)
\right]
\,dF(x).
\end{align*}
\end{theorem}

 We can use Theorem \ref{thm_sen2} to construct a normal approximation-based confidence interval for the Sen index using the newly proposed estimator. However, it is not easy to find a consistent estimator of the asymptotic variance.  This motivates us to develop EL and  JEL-based confidence interval for the Sen index. In this direction, first, we propose the empirical likelihood-based confidence interval for the Sen index in the following subsection, followed by JEL based confidence interval for the Sen index.
 \subsection{Empirical likelihood based confidence interval for the Sen index}
The empirical likelihood concept was initially introduced by Thomas and Grunkemier (1975) to ascertain the confidence interval for survival probabilities in the context of right censoring.   Owen's work (1988, 1990) extended the concept of empirical likelihood into a general technique. Empirical likelihood is determined by optimizing the non-parametric likelihood function under particular constraints.   Empirical likelihood provides a nonparametric framework that resembles likelihood methods.  Several researchers have investigated the application of empirical likelihood (EL) to enhance the robustness of inferences related to Lorenz curves, inequality indices, poverty gap differences, stochastic dominance tests, and other related analyses. Peng (2011), Qin (2013), Ratnasingam et al. (2024),  are various works which looks at application of EL in inequality studies.

Next, we propose an empirical likelihood (EL) based confidence interval for the Sen index.
Recall the estimator of the Sen index given in \eqref{senest}. For EL construction, we rewrite it as
\begin{equation}
\label{senest_el}
\widehat{S}_1 \;=\; \frac{2}{nqz} \sum_{i=1}^{n} (z - X_{(i)})(q-i)\, I\!\left(X_{(i)} \le z\right),
\end{equation}
where $q=n F(z)$.

Motivated by \eqref{senest_el}, we consider the estimating equation
\begin{equation}
\label{K_def}
K(X_i,S)
=
\left( 2(z - X_i)\big(F(z)-F(X_i)\big) - zS \right)
I(X_i \le z), \qquad i=1,\dots,n .
\end{equation}
It is easy to verify that $E\{K(X_i,S)\}=0$.

Let $X_1,\dots,X_n$ be a random sample from $F$. Let
$p=(p_1,\dots,p_n)'$ be a probability vector satisfying $p_i\ge 0$ and
$\sum_{i=1}^{n}p_i =1$.
Under these constraints, the product $\prod_{i=1}^{n} p_i$ attains its maximum value $n^{-n}$
at $p_i = 1/n$.
Therefore, the empirical likelihood ratio for $S$ is defined as
\begin{equation*}
L(S)
=
\sup_{\mathbf p}
\left\{
\prod_{i=1}^n (np_i)
:
\sum_{i=1}^n p_i = 1,\;
\sum_{i=1}^n p_i\, K(X_i,S)=0
\right\}.
\end{equation*}

Since $F$ is unknown, we replace it with the empirical distribution function $F_n$.
Define
\begin{equation*}
\widehat{K}(X_i,S)
=
\left( 2(z - X_i)\big(F_n(z)-F_n(X_i)\big) - zS \right)
I(X_i \le z), \qquad i=1,\dots,n.
\end{equation*}
Then the profile empirical likelihood ratio becomes
\begin{equation*}
\tilde{EL}(S)
=
\sup_{\mathbf p}
\left\{
\prod_{i=1}^n (np_i)
:
\sum_{i=1}^n p_i = 1,\;
\sum_{i=1}^n p_i\, \widehat{K}(X_i,S)=0
\right\}.
\end{equation*}

Using the method of Lagrange multipliers, the maximizer is
\begin{equation*}
p_i
=
\frac{1}{n}
\left(1 + \lambda_1 \widehat{K}(X_i,S)\right)^{-1},
\qquad i=1,\dots,n,
\end{equation*}
where the Lagrange multiplier $\lambda_1$ is the solution of
\begin{equation*}
\frac{1}{n}
\sum_{i=1}^{n}
\frac{ \widehat{K}(X_i,S) }
     { 1+\lambda_1 \widehat{K}(X_i,S) }
= 0.
\end{equation*}
Accordingly, the profile empirical likelihood ratio can be expressed as
\begin{equation*}
L(S)
=
\prod_{i=1}^{n}
\left(1 + \lambda_1\, \widehat{K}(X_i,S)\right)^{-1}.
\end{equation*}Under suitable regularity conditions, this representation facilitates the derivation of the asymptotic behavior of the empirical likelihood ratio, which is subsequently used to construct confidence intervals for the Sen index.\\
The limiting distribution of the empirical likelihood ratio is given in the following theorem.
\begin{theorem}
\label{thm:el}
Assume that $0 < F(z) < 1$ and $E[K(X,S)^2] < \infty$.
Then, as $n \to \infty$, $-2 \log L(S)$ converges in distribution to a
$\chi^2$ random variable with one degree of freedom.
\end{theorem}

Using Theorem~\ref{thm:el}, a $(1-\alpha)$ level EL-based confidence interval for $S$ is given by
\[
\left\{
S \;:\; -2\log L(S) \le \chi^2_{1}(1-\alpha)
\right\},
\] where $\chi^2_{(1-\alpha)}(1)$ is the $(1-\alpha)-$th percentile of the chi-square distribution with one degree of freedom.

Since the new estimator of the Sen index is a U-statistic, direct empirical likelihood inference is computationally challenging as it involves non-linear constraints and may exhibit poor finite-sample performance. The jackknife empirical likelihood (JEL) approach addresses these difficulties by utilising jackknife pseudo-values, resulting in simpler constraints and more stable inference. Accordingly, we adopt a JEL-based inference procedure for the Sen index.
In the following subsection, we detail the JEL-based inference for the Sen index.

\subsection{JEL based inference for Sen index}
When the constraints in the empirical likelihood  become nonlinear, evaluating the likelihood function becomes computationally challenging. Further, it becomes increasingly difficult as $n$ grows large. To overcome this problem with EL, Jing et al. (2009) introduced the jackknife empirical likelihood method for finding the confidence interval of the desired parametric function. This method is highly popular among researchers because it combines the effectiveness of the likelihood approach with the jackknife technique. Motivated by this, we construct a JEL-based confidence interval for $S$.
We define the estimating equation for constructing JEL as
 \begin{equation}\label{senjack}
   S_{n}=\frac{1}{\binom{n}{2}}\sum_{i=1}^{n}\sum_{j=1,j<i}^{n}\psi(X_i,X_j;\widehat{S})=0,
 \end{equation}
 where
  \begin{equation*}
 \psi(X_1,X_2)=2\psi_{1}(X_1,X_2)-z\widehat{S}.\psi_{2}(X_1,X_2).
  \end{equation*}
  Now we define the jackknife pseudo values for $S_{n}$ as
  \begin{equation*}
    \widehat{V}_{k}=nS_{n}-(n-1)S_{n-1,k};\,\,k=1,2,...,n,
  \end{equation*}
  where $S_{n-1,k}$ is calculated from (\ref{senjack}) using $X_1,X_2,...,X_{k-1},X_{k+1},...,X_{n}$. The jackknife pseudo-values are constructed so that their sample average exactly recovers the original U-statistic estimator.
  Thus we have (proof is given in the Appendix)
  \begin{equation*}
    \frac{1}{n}\sum_{k=1}^{n}\widehat{V}_{k}=S_{n}.
  \end{equation*} This property allows the U-statistic to be expressed as an average of approximately independent quantities, which is essential for implementing the jackknife empirical likelihood framework.
  Define the JEL ratio  function for the Sen index as
  \begin{equation*}
 JEL(S)=\sup_{\bf p} \big(\prod_{i=1}^{n}{(np_i)};\,\, \sum_{i=1}^{n}{p_i}=1;\,\,\sum_{i=1}^{n}{p_i \widehat{V}_{k}}=0\big).
\end{equation*}
By Lagrange multiplier method, the jackknife empirical log likelihood ratio is given by
\begin{equation*}
  J(S)=2\sum_{k=1}^{n}\log\big(1+\lambda_{2}\widehat{V}_{k}\big),
\end{equation*}
where $\lambda_{2}$ is the solution of
\begin{equation}\label{lambda2}
  \frac{1}{n}\sum_{k=1}^{n}{\frac{\widehat{V}_{k}}{1+\lambda_{2} \widehat{V}_{k}}}=0,
\end{equation}
    provided
\begin{equation*}
  \min_{{1\le k\le n}}\widehat{V}_{k}<\widehat{S}_n<  \max_{1\le k\le n}\widehat{V}_{k}.
\end{equation*}
We prove Wilks' theorem for jackknife empirical log likelihood ratio statistic and the result can be used to construct JEL based confidence interval for $S$.
\begin{theorem} \label{thm:senjel}
 Let $g(x)=E\left(\psi(X_1,X_2;S)|X_1=x\right)$ and $\sigma_{g}^{2}=Var(g(X))$.
  Suppose that $$E\left(\psi(X_1,X_2;S)\right)<\infty$$  and $\sigma_{g}^{2}>0$. Then as $n\rightarrow \infty$, $J(S)$ converges in distribution to $\chi^2$ with one degree of freedom.
\end{theorem}
Using Theorem \ref{thm:senjel} , we can construct  a $100(1-\alpha) \%$ JEL based confidence interval for $S$ as $$\left\{S|J(S)\le \chi^2_{1-\alpha}(1)\right\}.$$
\section{Inference for SST index}

As discussed in the introduction, Sen index was not additively decomposable, and hence Shorrocks (1995) reformulated Sen’s index to satisfy stronger axioms and to admit a clean multiplicative decomposition: incidence $\times$ average gap $\times$ (1+Gini of gaps). The “modified Sen” is known as the SST index and is widely used in policy diagnostics because each component is interpretable.

 The SST index is defined as

\begin{equation}\label{sstdef}
  S_{h}= 2\int_{0}^{z}\left(1-\frac{x}{z}\right)(1-F(x))dF(x).
\end{equation}

Based on the definition in the above equation (\ref{sstdef}), we develop a non-parametric estimator of the SST index in this section, along with ways to construct JEL and EL based confidence intervals for the SST index that will contribute towards reliable inferences for the SST index.

Let $F$ be the cumulative distribution function of the income under study. Suppose  $X_1,\ldots,X_n$ is random sample of size $n$ drawn from $F$.
The plug-in estimator for the SST index in (\ref{sstdef}) is given by
\begin{equation*}\label{sstemp}
\widehat S_{h1}=\frac{2}{zn^2}\sum_{k=1}^{q}(n-i)(z-X_{i}).
\end{equation*}
Xu (1998) examined the asymptotic distribution of $\widehat S_{h1}$ using the result from Stigler (1974).
The above estimator does not satisfy the population symmetry axiom. Hence, taking this into account,  Davidson(2009) had proposed the following estimator for the SST index
 \begin{equation}\label{qw}
   \tilde{S}_{h_{1}}=\frac{2}{n^{2}z}\sum_{i=1}^{q}(n-i+0.5)(z-X_{(i)}).
 \end{equation}
However, these estimators are not unbiased, and  therefore, we propose an unbiased estimator for the SST index using the theory of U-statistics.

\subsection{Unbiased estimator of SST index}

We propose a new U-statistic based estimator for SST index, which we later use to carry out  the JEL-based inference for SST index. Let $X_1$ and $X_2$ be independent copies from $F$.  Note that $\min(X_1,X_2)$ has a distribution function given by $(1-F^2(x))$.  To find an unbiased estimator of the SST index, we rewrite  equation (\ref{sstdef}) as
\begin{eqnarray*}
S_{h}&=& \frac{2}{z}\int_{0}^{z}(z-x)(1-F(x))dF(x)\\
&=&\int_{0}^{z}2(1-F(x))dF(x)-\frac{1}{z}\int_{0}^{z}2x(1-F(x))dF(x)\\
&=&\int_{0}^{z}2(1-F(x))dF(x)-\frac{1}{z}\int_{0}^{\infty}2xI(x\le z)(1-F(x))dF(x)\\
&=&P\big(\min(X_1,X_2) \le z\big)-\frac{1}{z}E\big(\min(X_1,X_2)I\big(\min(X_1,X_2) \le z\big)\big).
\end{eqnarray*}
Let  $\psi_3(X_1,X_2))$ be a kernel of degree 2 given as
\begin{equation*}
\psi_3(X_1,X_2))=I\big(\min(X_1,X_2) \le z\big)-\frac{1}{z}\min(X_1,X_2)I\big(\min(X_1,X_2) \le z\big),
\end{equation*}
then $E(\psi_3(X_1,X_2))) = S_h$.  Thus, an unbiased  estimator of $S_{h}$ based on  U-statistics is given by
\begin{equation}\label{23}
\widehat{S}_h=\frac{1}{\binom{n}{2}}\sum_{i=1}^{n}\sum_{j=1,j<i}^{n}\psi_3(X_i,X_j).
\end{equation}

Next, we simplify the expression in (\ref{23}) using order statistics. Note that, for a fixed $k$ such $0 \leq k \leq n$, we can consider the set of ordered pairs $(i,j)$ such that $i<j$ and $\min(X_i,X_j) =X_{(k)}$. It can easily be seen that these are pairs with $X_{(k)}$ and any later observations  $X_{(k+1)},X_{(k+2}, \ldots, X_{(n)}$ and hence there will be exactly there will be $n-k$ ordered pairs with $X_{(k)}$ as minimum value. Thus, \eqref{23} can be rewritten as
\begin{equation*}
\widehat{S}_h=\frac{2}{n(n-1)}\sum_{i=1}^{n}{(n-i)I(X_{(i)} \le z)\big(\frac{z-X_{(i)}}{z}\big)}.
\end{equation*}
Since $I(X_{(i)}\le z)=0$ for $i=q+1,q+2,...,n$, the above equation further
simplified to
\begin{equation}\label{sstpract}
\widehat{S}_h=\frac{2}{n(n-1)z}\sum_{i=1}^{q}{(n-i)(z-X_{(i)})}.
\end{equation}Being a U-statistic, $\widehat{S}_h$ is a consistent estimator of  $S_{h}.$  That is, $\widehat{S}_h$ converges in probability to $S_{h}$, as $n\rightarrow\infty$ (Lehmann, 1951).

\noindent Next we find the asymptotic distribution of $\widehat{S}_h$ and the proof is given in Appendix.
\begin{theorem}
   Given $z$, as $n\rightarrow \infty$, $\sqrt{n}(\widehat{S}_h-S_h)$ converges in distribution to Gaussian with mean zero and variance $\sigma^2=4\sigma_2^2$ where $\sigma_2^2$ is given by
  \begin{eqnarray*}
  \sigma_2^2&=&Var\Big(\big(F(X)-\frac{1}{z}\int_{0}^{X}ydF(y)\big)I(X>z)\\&&\quad \quad +
  \big(1-\frac{1}{z}X+F(X)-\frac{1}{z}(X+\int_{0}^{X}ydF(y))\big)I(X\le z)\Big).
  \end{eqnarray*}
   \end{theorem}
\subsection{Empirical likelihood based confidence interval for the SST index}

   Similar to the approach detailed in Section 2.2, given a random sample $X_1,X_2,...,X_n$ of size $n$ from $F(.)$,  the empirical likelihood function for $S_{h}$ can be written as
\begin{equation*}
  {EL}(S_{h})=\sup_{\bf p} \big(\prod_{i=1}^{n}{p_i};\sum_{i=1}^{n}{p_i}=1;\sum_{i=1}^{n}{p_{i} {M}(X_i,S_{h})}=0\big),
\end{equation*}
where ${M}(X_i,S_{h})=2(z-X)(1-F(X))-zS_{h}$.

Hence, the profile empirical likelihood ratio for $S_{h}$ can be written as
\begin{equation*}
  \widetilde{EL}(S_{h})=\sup_{\bf p} \big(\prod_{i=1}^{n}{(np_i)};\sum_{i=1}^{n}{p_i}=1;\sum_{i=1}^{n}{p_{i} \widehat{M}(X_i,S_{h})}=0\big),
\end{equation*}
where $\widehat{M}(X_i,S_{h})=2(z-X_i)(1-F_{n}(X_i))-zS_{h}$, $i=1,2,\ldots,n$.
The empirical log likelihood ratio can be  derived using the Lagrange multiplier method and is obtained as
   $$J_1(S_{h})=-2\log  \widetilde{EL}(S_{h})=2\sum_{i=1}^{n}\log\big[1+\lambda_{3} \widehat{M}(X_i,S_{h})\big],$$
where $\lambda_{3}$ is the solution to the equation
   $$\frac{1}{n}\sum_{i=1}^{n}{\frac{\widehat{M}(X_i,S_{h})}{1+\lambda_{3} \widehat{M}(X_i,S_{h})}}=0.$$
The following theorem gives the limiting distribution of $J(S_{h})$ and the proof is similar to that of Theorem \ref{thm:el} and hence omitted.
\begin{theorem}\label{thm_sstel}
 Assume  $E(X^2)<\infty$, then as $n\rightarrow \infty$, $J_1(S_{h})$ converges in distribution to a  $\chi^2$ random variable with one degree of freedom.
\end{theorem}

Using Theorem \ref{thm_sstel}, we construct a $(1-\alpha)$ level empirical likelihood  based confidence interval for $S_h$  as
\begin{equation*}
  \big(S_h|J_1(S_h) \le\chi_{\alpha}^2(1)\big).
\end{equation*}

 \subsection{JEL based inference for SST index}
We can obtain a JEL based confidence interval for $S_h$  similar to that of $S$ as discussed in subsection 2.3.
To develop JEL-based inference for $S_h$ we make use of the U-statistics-based estimator given  in subsection 3.1.  We define jackknife pseudo values for $S_{h}$ as
  \begin{equation*}
    \widehat{Q}_{k}=n\widehat{S}_h-(n-1)\widehat{S}_{h,n-1,k};\,\,k=1,2,...,n,
  \end{equation*}
  where $\widehat{S}_{h,n-1,k}$ is calculated from (\ref{23}) using $X_1,X_2,...,X_{k-1},X_{k+1},...,X_{n}$.
  Define the jackknife empirical likelihood ratio  for $S_{h}$  as
  \begin{equation*}
 J_2(S_h)=\sup_{\bf p} \big(\prod_{i=1}^{n}{(np_i)};\,\, \sum_{i=1}^{n}{p_i}=1;\,\,\sum_{i=1}^{n}{p_i \widehat{Q}_{k}}=0\big).
\end{equation*}
By Lagrange multiplier method, we obtain the jackknife empirical log likelihood ratio  as
\begin{equation*}
  J_2(S_h)=2\sum_{k=1}^{n}\log\big(1+\lambda_{4}\widehat{Q}_{k}\big),
\end{equation*}
where $\lambda_{4}$ is the solution of
   $$\frac{1}{n}\sum_{k=1}^{n}{\frac{\widehat{Q}_{k}}{1+\lambda_{4} \widehat{Q}_{k}}}=0,$$
    provided
\begin{equation}\label{lamb11}
  \min_{{1\le k\le n}}\widehat{Q}_{k}<\widehat{S}_h<  \max_{1\le k\le n}\widehat{Q}_{k}.
\end{equation}Next, we obtain the asymptotic distribution of the jackknife empirical log likelihood ratio and the result is stated in the following theorem, and the proof is direct from Theorem 1 of Jing et al. (2009) for the appropriate choice of the kernel.
   \begin{theorem}\label{thm:sstjel} Let $g_2(x)=E\left(\psi_3(X_1,X_2)|X_1=x\right)$ and $\sigma_{g}^{2}=Var(g_2(X))$.
  Assume that   $\sigma_{g}^{2}>0$ and $E\left(\psi_3(X_1,X_2)\right)<\infty$.  Then, as $n\rightarrow \infty$, $J_2(S_h)$ converges in distribution to $\chi^2$ a random variable with one degree of freedom.
\end{theorem}

Using Theorem \ref{thm:sstjel}, we construct a $(1-\alpha)$ level JEL  based confidence interval for $S_h$  as
\begin{equation*}
  \Big(S_h|J_2(S) \le\chi_{\alpha}^2(1)\Big).
\end{equation*}

\section{Simulation results}

In this section, a comprehensive Monte Carlo study has been conducted to compare the proposed estimators and confidence intervals outlined in Sections 2 and 3.  First, we compare the U-statistics based estimator obtained in Sections 2 and 3 with the other available estimators of the Sen and SST indices. We calculate the bias and mean squared error of these estimators for comparison.  Subsequently, we compare the proposed EL and JEL-based confidence intervals with the confidence interval constructed using the approach outlined in Davidson (2009). The performance of the confidence intervals is evaluated based on coverage probability and average width. The simulation is carried out in R and repeated ten thousand times.

\subsection{Comparison of Estimators}
The performance of the proposed estimators of Sen is compared with the plug-in estimator, while the proposed estimator of the SST index is evaluated against both the plug-in estimator and the one introduced by Davidson (2009), as presented in equation (\ref{sendav}). The bias and mean squared error of these estimators are determined for various distributions. In the simulation, we specified the value of $z$ as $1.41$.

First, we simulate observations from an exponential distribution. We simulate observations using different sample sizes ($n=20,40,60,80,100$). The bias and MSE obtained for the Sen index and SST index using the estimators given in (\ref{sendav}) and (\ref{senustat}) are reported in Tables \ref{tab_sen_exp} and \ref{tab_sst_exp}.
\begin{table}
\caption{Sen index comparison: Exponential distribution ($\lambda$)}
\label{tab_sen_exp}
\centering
\scalebox{0.8}{
\begin{tabular}{|c|c|c|c|c|c|}\hline
&&\multicolumn{2}{c|}{Proposed estimator} & \multicolumn{2}{c|}{Plug-in estimator} \\ \cline{3-6}
$\lambda$&$n$ & Bias ($\times 10^{-2})$& MSE ($\times 10^{-2}$) & Bias  ($\times 10^{-2})$& MSE ($\times 10^{-2}$) \\ \hline
\multirow{5}{*}{2}&20&0.241 & 0.201& -0.672& 0.221	\\
&40& 0.167&0.180&-0.451&0.184\\
&60&-0.152&0.109& -0.374&0.111	\\
&80& -0.141&0.087&-0.243	&0.088\\
&100&-0.084&0.067 &-0.166	&0.068\\ \hline
\multirow{5}{*}{4}&20& 0.197& 0.042& -0.349&0.046	\\
&40&0.128	&0.033&-0.276 &0.034\\
&60&-0.107&0.024& -0.244&0.025	\\
&80&0.067&0.014 &	-0.091&0.014\\
&100&-0.043& 0.011&	-0.067&0.013\\ \hline
\multirow{5}{*}{6}&20& -0.094& 0.021& -0.280&	0.024\\
&40&	-0.061&0.013& -0.240&0.014\\
&60&-0.043&0.010& -0.126&	0.012\\
&80&0.012& 0.010&	-0.090&0.011\\
&100&-0.010& 0.009&	-0.043&0.010\\ \hline
\end{tabular}}
\end{table}

\begin{table}
\centering
\caption{SST index comparison: Exponential distribution ($\lambda$) }
\label{tab_sst_exp}
\scalebox{0.8}{
\begin{tabular}{|c|c|c|c|c|c|c|c|}\hline
&&\multicolumn{2}{c|}{Proposed estimator} & \multicolumn{2}{c|}{Shorrock estimator}&\multicolumn{2}{c|}{Davidson estimator} \\ \cline{3-8}
$\lambda$ &$n$ & Bias ($\times 10^{-2})$& MSE ($\times 10^{-2})$& Bias ($\times 10^{-2})$ & MSE ($\times 10^{-2})$ & Bias ($\times 10^{-2})$ & MSE ($\times 10^{-2})$\\ \hline
\multirow{5}{*}{2}&20& 0.111& 0.203& -0.669&0.213& -0.669&	0.213\\
&40&-0.137& 0.104& -0.529&0.108	& -0.529&0.108\\
&60&0.078& 0.068& -0.182&0.069	& -0.182&0.069	\\
&80& 0.071& 0.054& -0.123&0.055	& -0.123&	0.055	\\
&100& -0.029& 0.041& -0.186&0.042	& -0.018&	0.042\\ \hline
\multirow{5}{*}{4}&20&0.171 &0.051 &-0.315 &	0.063& -0.315&	0.063\\
&40&-0.101 &0.041 &-0.275 &	0.047& -0.275&	0.047\\
&60&0.080 &0.024 &-0.105&	0.026& -0.105&	0.026\\
&80&0.042 &0.013 &-0.068 &	0.018& -0.068&	0.018\\
&100&0.021 &0.011 &-0.034 &	0.016& -0.034&	0.016\\ \hline
\multirow{5}{*}{4}&20&0.190 &0.067 &-0.234 &	0.071& -0.234&	0.071\\
&40&-0.091 &0.056 &-0.171 &	0.057& -0.171&	0.056\\
&60&0.086 &0.032 &-0.111&	0.033& -0.111&0.033\\
&80&0.064 &0.026 &-0.074 &	0.029& -0.074&	0.029\\
&100&0.039 &0.010 &-0.053 &	0.012& -0.053&	0.012\\\hline
\end{tabular}}
\end{table}

Next, we simulate observations from the Pareto distribution using various values for the scale parameter $k$ and the shape parameter $\alpha$. The use of the Pareto distribution is appropriate as it is frequently used to model income distribution.  The bias and mean squared error for the Sen index and SST index are presented in Table \ref{tab_sen_pareto} and Table \ref{tab_sst_pareto}, respectively.
\begin{table}
\centering
\caption{Sen index comparison: Pareto distribution $(k,\alpha)$ }
\label{tab_sen_pareto}
\scalebox{0.8}{
\begin{tabular}{|c|c|c|c|c|c|}\hline
&&\multicolumn{2}{c|}{Proposed estimator} & \multicolumn{2}{c|}{Plug-in estimator} \\ \cline{3-6}
$(k,\alpha)$ &$n$ & Bias ($\times 10^{-2})$& MSE ($\times 10^{-2})$ & Bias  ($\times 10^{-2})$& MSE ($\times 10^{-2})$  \\ \hline
\multirow{5}{*}{(1,1)}&20&-7.341&0.543&-7.611&0.546\\
&40&-3.231&0.311&-3.412&0.342\\
&60&2.413&	 0.287&2.821&0.279\\
&80&1.812&0.181&1.861&0.192\\
&100&1.523&	0.151&1.561& 0.173 \\\hline
\multirow{5}{*}{(1,2)}&20& -6.067 & 0.071& -0.675&0.079	\\
&40& 3.067	&0.041 &3.876& 0.041\\
&60&0.000&0.020& -0.086&	0.020\\
&80&0.000& 0.011&-0.054	&0.010\\
&100&0.000& 0.010&-0.041	&0.010\\ \hline

\multirow{5}{*}{(1,3)}&20&0.341 &0.071 &0.381 &	0.072\\
&40&0.125&0.061&0.123 &0.061\\
&60&0.007&0.042& 0.007&0.043	\\
&80&0.001&0.020	&0.002&0.020\\
&100& 0.000&0.010&	0.000&0.010\\ \hline
\end{tabular}}
\end{table}

\begin{table}
\centering
\caption{SST index comparison: Pareto distribution $(k,\alpha)$}
\label{tab_sst_pareto}
\scalebox{0.8}{
\begin{tabular}{|c|c|c|c|c|c|c|c|}
\hline
&&\multicolumn{2}{c|}{Proposed estimator} & \multicolumn{2}{c|}{Shorrock estimator} & \multicolumn{2}{c|}{Davidson estimator} \\ \cline{3-8}
$(k,\alpha)$& $n$ & Bias ($\times 10^{-2}$) & MSE ($\times 10^{-2}$) & Bias ($\times 10^{-2}$) & MSE ($\times 10^{-2}$) & Bias ($\times 10^{-2}$) & MSE ($\times 10^{-2}$) \\ \hline

\multirow{5}{*}{(1,1)} & 20 & 0.089 & 0.045 & -0.091 & 0.046 & -0.094 & 0.046 \\
& 40 & 0.078 & 0.033 & -0.085 & 0.034 & -0.085 & 0.034 \\
& 60 & 0.051 & 0.031 & -0.067& 0.031 & -0.067 & 0.032 \\
& 80 & 0.023 & 0.011 & -0.025 & 0.012 & -0.026 & 0.012 \\
& 100 & 0.011 & 0.011 & -0.016 & 0.011 & -0.017 & 0.011 \\ \hline

\multirow{5}{*}{(1,2)} & 20 & 0.041 & 0.095 & -0.276 & 0.094 & -0.277 & 0.094 \\
& 40 & 0.021 & 0.035 & -0.062 & 0.035 & -0.062 & 0.035 \\
& 60 & 0.011 & 0.013 & -0.042 & 0.016 & -0.042 & 0.017 \\
& 80 & 0.008 & 0.013 & -0.029 & 0.013 & -0.029 & 0.013 \\
& 100 & 0.004 & 0.009 & -0.012 & 0.011 & -0.012 & 0.011 \\ \hline

\multirow{5}{*}{(1,3)} & 20 & 0.050 & 0.075 & -0.250 & 0.075 & -0.250 & 0.075 \\
& 40 & 0.045 & 0.033 & -0.195 & 0.034 & -0.191 & 0.034 \\
& 60 & 0.011 & 0.023 & -0.089 & 0.023 & -0.089 & 0.023 \\
& 80 & 0.006 & 0.017 & -0.013 & 0.017 & -0.013 & 0.017 \\
& 100 & 0.005 & 0.015 & -0.011 & 0.015 & -0.011 & 0.015 \\ \hline
\end{tabular}
}
\end{table}

Finally, we generated observations from a heavy-tailed distribution, specifically a log-normal distribution, for parameter choices for $\mu$ and $\sigma$. The bias and MSE calculated for the Sen and SST indices are presented in Tables \ref{tab_sen_lnorm} and \ref{tab_sst_lnorm}, respectively.
\begin{table}
\centering
\caption{Sen index comparison: Log normal distribution $(\mu,\sigma)$}
\label{tab_sen_lnorm}
\scalebox{0.8}{
\begin{tabular}{|c|c|c|c|c|c|}
\hline
&&\multicolumn{2}{c|}{Proposed estimator} & \multicolumn{2}{c|}{Plug-in estimator} \\ \cline{3-6}
$(\mu,\sigma)$&$n$ & Bias ($\times 10^{-2})$& MSE($\times 10^{-2})$ & Bias  ($\times 10^{-2})$& MSE ($\times 10^{-2})$  \\ \hline
\multirow{5}{*}{(0,1)} &20&-7.341&0.544&-7.611&0.544\\
&40&-3.231&0.311&-3.412&0.340\\
&60&2.413&0.283&2.821&0.279\\
&80&1.812&0.188&1.861&0.191\\
&100&1.523&0.157&1.561&0.170\\ \hline
\multirow{5}{*}{(1,1)} &20& -6.325& 0.616& -6.761&0.617\\
&40& -4.130&0.318&-4.786 &0.423\\
&60& 3.123&0.287& 3.451&0.301\\
&80&1.987&0.175& -2.123&0.180\\
&100&0.934&0.130& 1.112&0.130\\ \hline
\multirow{5}{*}{(1,2)} &20&-2.841& 0.877&-3.650&0.901\\
&40&-1.326&0.841&-1.981&0.867\\
&60&-0.926&0.612&-1.121&0.612\\
&80&-0.671&0.452&-0.871&0.501\\
&100&-0.431&0.278&-0.461&0.302\\ \hline
\end{tabular}
}
\end{table}

\begin{table}
\centering
\caption{SST index comparison: Log normal distribution $(\mu,\sigma)$}
\label{tab_sst_lnorm}
\scalebox{0.8}{
\begin{tabular}{|c|c|c|c|c|c|c|c|}\hline
&&\multicolumn{2}{c|}{Proposed estimator} & \multicolumn{2}{l|}{Shorrock estimator}&\multicolumn{2}{c|}{Davidson estimator} \\ \cline{3-8}
$(\mu,\sigma)$&$n$ & Bias($\times 10^{-2})$ & MSE ($\times 10^{-2})$& Bias  ($\times 10^{-2})$ & MSE ($\times 10^{-2})$ & Bias ($\times 10^{-2})$ & MSE ($\times 10^{-2})$ \\ \hline
\multirow{5}{*}{(0,1)}&20&0.077 &0.758 &-0.821 &	0.758& -0.821&	0.758\\
&40&0.057 &0.691 &-0.627 &	0.699& -0.627&	0.699\\
&65&0.039 &0.406 &-0.541 &	0.422& -0.541&	0.422\\
&80&0.028 &0.348 &-0.271 &	0.348& -0.271&	0.348\\
&100&0.010 &0.296 &-0.097 &	0.298& -0.097&	0.298\\  \hline
\multirow{5}{*}{(1,1)}&20 &-0.271& 0.606& -0.374&0.611	&-0.374 &0.611	\\
&40&-0.133& 0.304& -0.346&0.346	&-0.346&0.346	\\
&60&-0.033& 0.202& -0.117&0.214	&-0.117 &0.214	\\
&80&-0.021& 0.125& -0.097&0.127	&-0.097 &0.127	\\
&100&-0.010& 0.101& -0.066&0.101	&-0.066 &0.101	\\ \hline
\multirow{5}{*}{(1,2)}&20& -0.191& 1.321& -1.024&1.356	&-1.024 &1.356	\\
&40& -0.103& 0.911& -0.901&0.911	&-0.901&0.911	\\
&60& -0.076& 0.712& -0.752&0.763	&-0.752&0.763	\\
&80& -0.042& 0.652& -0.581&0.652	&-0.581&0.652	\\
&100& -0.028& 0.594& -0.109&0.595	&-0.109&0.595\\  \hline
\end{tabular}
}
\end{table}

It is observed that for the Sen index, the proposed estimator exhibits a lower bias compared to the plug-in estimator, while the mean squared errors of both estimators remain comparable across the three distributions under consideration.   The proposed estimator for the SST index demonstrates reduced bias relative to the plug-in estimator (Shorrock estimator in the table) and the Davidson estimator, while the mean squared errors of all three estimators remain comparable.

\subsection{Comparison of confidence intervals}
The proposed EL and JEL-based confidence intervals were compared with the confidence interval suggested by Davidson (2009), utilizing normal approximation for both the Sen and SST indices. As in the previous subsection, in the simulation study for confidence interval, we considered the same three distributions: exponential, Pareto, and lognormal distributions.

 The  coverage probability and average length of the estimated confidence intervals of the Sen index for each distribution are presented in Tables \ref{tab_ci_sen_exp}, \ref{tab_ci_sen_pareto}, and \ref{tab_ci_sen_lnorm}, respectively.

\begin{table}
\centering
\caption{Comparison of confidence intervals of Sen index: Exponential distribution $(\lambda)$ }
\label{tab_ci_sen_exp}
\scalebox{0.9}{
\begin{tabular}{|c|c|c|c|c|c|c|c|}
\hline
 & & \multicolumn{2}{c|}{$\lambda=2$} & \multicolumn{2}{c|}{$\lambda=4$} & \multicolumn{2}{c|}{$\lambda=6$} \\ \cline{3-8}
$n$ & Method & CP & AL & CP & AL & CP & AL \\ \hline

\multirow{3}{*}{20}
 & EL        & 0.930 & 0.521 & 0.930 & 0.503 & 0.930 & 0.497 \\
 & JEL       &   0.912    & 0.575      & 0.903 & 0.484 & 0.903 & 0.353 \\
 & Davidson  & 0.900 & 0.228 & 0.970 & 0.191 & 0.970 & 0.189 \\ \hline

\multirow{3}{*}{40}
 & EL        & 0.910 & 0.379 & 0.920 & 0.365 & 0.930 & 0.360 \\
 & JEL       &  0.921     &  0.463     & 0.908 & 0.481 & 0.908 & 0.334 \\
 & Davidson  & 0.920 & 0.162 & 0.930 & 0.105 & 0.940 & 0.101 \\ \hline

\multirow{3}{*}{60}
 & EL        & 0.940 & 0.309 & 0.940 & 0.297 & 0.940 & 0.293 \\
 & JEL       &   0.949    &   0.312    & 0.924 & 0.341 & 0.924 & 0.321 \\
 & Davidson  & 0.970 & 0.129 & 0.945 & 0.099 & 0.960 & 0.091 \\ \hline

\multirow{3}{*}{80}
 & EL        & 0.965 & 0.268 & 0.970 & 0.257 & 0.970 & 0.254 \\
 & JEL       &   0.950    &  0.276     & 0.936 & 0.302 & 0.920 & 0.311 \\
 & Davidson  & 0.930 & 0.112 & 0.937 & 0.079 & 0.960 & 0.085 \\ \hline

\multirow{3}{*}{100}
 & EL        & 0.945 & 0.241 & 0.930 & 0.232 & 0.940 & 0.229 \\
 & JEL       &    0.949   &   0.234    & 0.949 & 0.251 & 0.947 & 0.290 \\
 & Davidson  & 0.914 & 0.102 & 0.951 & 0.056 & 0.930 & 0.061 \\ \hline
\end{tabular}}
\end{table}

 \begin{table}
\centering
\caption{Comparison of confidence interval of Sen index: Pareto distribution ($k,\alpha$)}
\label{tab_ci_sen_pareto}
\scalebox{0.9}{
\begin{tabular}{|l|l|l|l|l|l|l|l|}
\hline
&&\multicolumn{2}{c|}{($k=1,\alpha=1 $)}&\multicolumn{2}{c|}{($k=1,\alpha=2 $)}&\multicolumn{2}{c|}{($k=1,\alpha=3$)}\\\cline{3-8}
$n$                   & Interval & CP & AL& CP & AL& CP & AL \\ \hline
\multirow{3}{*}{20} &EL    &  0.870& 0.110 & 0.920 &0.138& 0.950 &  0.149       \\
                    & JEL     &0.920 & 0.223 &0.920  &0.116&  0.920&  0.077        \\
                    & Davidson &  0.710 & 0.052 & 0.790 &0.075  & 0.890&  0.083    \\ \hline
\multirow{3}{*}{40} &EL  & 0.920  & 0.082 & 0.920 & 0.101 &0.950& 0.109                  \\
                    & JEL   &0.937 & 0.143 &0.942  &0.072&  0.934&  0.048        \\
                    & Davidson &  0.780  & 0.046 & 0.870 &0.063 & 0.910&  0.067     \\ \hline
\multirow{3}{*}{60} & EL  &  0.960 & 0.067 & 0.940 & 0.083 &0.930&  0.089 \\
                    & JEL   &0.947 & 0.108 &0.944 &0.058&  0.944&  0.036        \\
                    & Davidson &  0.870  & 0.043 & 0.880 &0.055 & 0.880&  0.057    \\ \hline
\multirow{3}{*}{80} & EL   & 0.940& 0.059 & 0.960 &  0.072&0.957& 0.057             \\
                   & JEL   &0.940 & 0.099 &0.939 &0.050&  0.946&  0.033       \\
                    & Davidson &  0.890  & 0.039 & 0.950 &0.049 & 0.945&  0.033    \\ \hline
\multirow{3}{*}{100} & EL    & 0.900& 0.053 &0.910  &0.065  &0.961& 0.042   \\
                   & JEL   &0.956 & 0.084 &0.957 &0.042&  0.957&  0.027       \\
                    & Davidson &  0.900  & 0.037 & 0.920 &0.045 & 0.959&  0.027   \\ \hline

\end{tabular}}
\end{table}

\begin{table}
\centering
\caption{Comparison of confidence interval  of Sen index:  Lognormal distribution $(\mu,\sigma)$ }
\label{tab_ci_sen_lnorm}
\scalebox{0.9}{
\begin{tabular}{|l|l|l|l|l|l|l|l|}
\hline
&&\multicolumn{2}{c|}{$(\mu,\sigma)=(0,1)$}&\multicolumn{2}{c|}{$(\mu,\sigma)=(1,1)$}&\multicolumn{2}{c|}{$(\mu,\sigma)=(1,2)$}\\\cline{3-8}
$n$                   & Interval & CP & AL& CP & AL& CP & AL \\ \hline
\multirow{3}{*}{20} &EL& 0.930 & 0.455 &0.870  &0.266&  0.920&     0.434    \\
                    & JEL & 0.908   &  0.276&  0.914&  0.779&   0.931    &   0.713 \\
                    & Davidson     & 0.910 & 0.308 & 0.810 &0.177& 0.890 &   0.151  \\ \hline
\multirow{3}{*}{40} &EL  & 0.960  &  0.336& 0.900 & 0.211 &0.960& 0.326     \\
                    & JEL    & 0.924 &0.197  & 0.941 &0.561  &0.941&    0.517           \\
                    & Davidson      & 0.930  & 0.227 &0.920  & 0.151 &0.930& 0.124    \\ \hline
\multirow{3}{*}{60} & EL  &  0.910 & 0.265 & 0.940 &0.168  &0.950& 0.109  \\
                    & JEL    &   0.939& 0.156 & 0.945 & 0.423 &0.947&        0.398       \\
                    & Davidson       & 0.920 &  0.184& 0.930 &0.124&  0.900&      0.194        \\ \hline
\multirow{3}{*}{80} & EL   & 0.960& 0.236 & 0.940 & 0.146 &0.950&         0.227     \\
                    & JEL    & 0.949 & 0.135 & 0.961 & 0.272 &0.951&    0.271\\
                    & Davidson  & 0.950& 0.161 & 0.930 &  0.110& 0.930&0.172            \\ \hline
\multirow{3}{*}{100} & EL    &0.970 & 0.214 &0.950  & 0.133 &0.950& 0.208   \\
                    & JEL    &0.0946 & 0.119 &0.941  &0.198  &0.950& 0.185   \\
                    & Davidson      & 0.950&  0.145&  0.910& 0.100 &0.970&           0.156      \\ \hline

\end{tabular}}
\end{table}
The above tables indicate that both EL and JEL-based confidence intervals exhibit better coverage probability compared to the normal distribution-based confidence intervals proposed by Davidson. The performance of EL and JEL-based confidence intervals is better even with small sample sizes.

Next, in Tables \ref{tab_ci_exp}, \ref{tab:ci_sst_pareto}, and \ref{tab_ci_sst_lnorm}, we present the coverage probability and average width of the proposed EL and JEL confidence intervals for the SST index, comparing them with the normal distribution-based confidence intervals suggested by Davidson.

\begin{table}
\centering
\caption{Comparison of confidence interval of SST index: Exponential distribution $(\lambda)$ }
\label{tab_ci_exp}
\scalebox{0.9}{
\begin{tabular}{|l|l|l|l|l|l|l|l|}
\hline
&&\multicolumn{2}{c|}{$\lambda=2$}&\multicolumn{2}{c|}{$\lambda=4$}&\multicolumn{2}{c|}{$\lambda=6$}\\\cline{3-8}
$n$                   & Interval & CP & AL& CP & AL& CP & AL \\ \hline
\multirow{3}{*}{20} &EL    & 0.946 & 0.438 &0.947  &0.529& 0.952 &     0.517    \\
                    & JEL     &0.920 & 0.223 &0.920  &0.116&  0.920&  0.077        \\
                    & Davidson &  0.923  & 0.223 & 0.923 &0.116  & 0.925&  0.077     \\ \hline
\multirow{3}{*}{40} &EL  &  0.947 &  0.341& 0.951 & 0.436 & 0.961&       0.373            \\
                    & JEL   &0.937 & 0.143 &0.942  &0.072&  0.934&  0.048        \\
                    & Davidson &  0.942  & 0.147 & 0.940 &0.072 & 0.940&  0.048     \\ \hline
\multirow{3}{*}{60} & EL  & 0.946  & 0.317 &0.948  &0.342  & 0.951&  0.341 \\
                    & JEL   &0.947 & 0.108 &0.944 &0.058&  0.944&  0.036        \\
                    & Davidson &  0.951  & 0.114 & 0.945 &0.058 & 0.944&  0.037     \\ \hline
\multirow{3}{*}{80} & EL  & 0.953 &0.276 & 0.949 &  0.289 & 0.949 &   0.276          \\
                   & JEL   &0.940 & 0.099 &0.939 &0.050&  0.946&  0.033       \\
                    & Davidson &  0.940  & 0.092 & 0.940 &0.050 & 0.945&  0.033    \\ \hline
\multirow{3}{*}{100} & EL    &0.943 &  0.247&0.946  & 0.244 & 0.948&  0.246  \\
                   & JEL   &0.956 & 0.084 &0.957 &0.042&  0.957&  0.027       \\
                    & Davidson &  0.960  & 0.080 & 0.959 &0.041 & 0.959&  0.027   \\ \hline

\end{tabular}}
\end{table}

\begin{table}
\centering
\caption{Comparison of confidence interval of SST index: Pareto  distribution $(k,\alpha)$ }
\label{tab:ci_sst_pareto}
\scalebox{0.9}{
\begin{tabular}{|l|l|l|l|l|l|l|l|}
\hline
&&\multicolumn{2}{c|}{$(k,\alpha)=(1,1)$}&\multicolumn{2}{c|}{$(k,\alpha)=(1,2)$}&\multicolumn{2}{c|}{$(k,\alpha)=(1,3)$}\\\cline{3-8}
$n$                   & Interval & CP & AL& CP & AL& CP & AL \\ \hline
\multirow{3}{*}{20} &EL    & 0.943 & 0.342 &  0.941&0.241& 0.907 & 0.232        \\
                    & JEL     &0.933 & 0.126 &0.939  &0.131&  0.917&  0.117       \\
                    & Davidson &  0.866  & 0.102 & 0.893 &0.104  & 0.887&  0.093   \\ \hline
\multirow{3}{*}{40} &EL  &  0.942 & 0.237 &0.938  & 0.110 &0.945&0.   109                \\
                    & JEL   &0.941& 0.092 &0.943  &0.089&  0.947&  0.103       \\
                    & Davidson &  0.911  & 0.090 & 0.947&0.083 & 0.937&  0.070    \\ \hline
\multirow{3}{*}{60} & EL  &0.950   & 0.081 &  0.949& 0.079 &0.951&0.077  \\
                    & JEL   &0.946 & 0.075 &0.952 &0.070&  0.949&  0.067       \\
                    & Davidson &  0.929  & 0.067 & 0.938&0.066& 0.942&  0.058    \\ \hline
\multirow{3}{*}{80} & EL   &0.943 & 0.067 & 0.944 &0.059  &0.951& 0.064             \\
                   & JEL   &0.945& 0.061 &0.943 &0.061&  0.952&  0.052    \\
                    & Davidson &  0.912 & 0.058& 0.931 &0.058 & 0.945&  0.053   \\ \hline
\multirow{3}{*}{100} & EL    & 0.939& 0.151 &0.945  & 0.055 &0.950&   0.050 \\
                   & JEL   &0.948 & 0.048 &0.942 &0.052&  0.947&  0.048     \\
                    & Davidson &  0.935  & 0.053& 0.931 &0.052 & 0.938&  0.046   \\ \hline

\end{tabular}}
\end{table}

\begin{table}
\centering
\caption{Comparison of confidence interval of SST index: Lognormal Distribution}
\label{tab_ci_sst_lnorm}
\scalebox{0.9}{
\begin{tabular}{|l|l|l|l|l|l|l|l|}
\hline
&&\multicolumn{2}{c|}{$(\mu,\sigma)=(0,1)$}&\multicolumn{2}{c|}{$(\mu,\sigma)=(1,1)$}&\multicolumn{2}{c|}{$(\mu,\sigma)=(1,2)$}\\\cline{3-8}
$n$                   & Interval & CP & AL& CP & AL& CP & AL \\ \hline
\multirow{3}{*}{20} &EL    &0.941  & 0.367 & 0.913 &0.312& 0.940 &0.421         \\
                    & JEL     &0.927 & 0.369 &0.904  &0.302&  0.939&  0.442        \\
                    & Davidson &  0.930  & 0.364 & 0.872 &0.279  & 0.901&  0.417     \\ \hline
\multirow{3}{*}{40} &EL  &  0.931 & 0.286 & 0.940 &0.297  &0.935&0.311                   \\
                    & JEL   &0.930& 0.247 &0.938  &0.212&  0.936&  0.312       \\
                    & Davidson &  0.929  & 0.231 & 0.916 &0.210 & 0.935&  0.306    \\ \hline
\multirow{3}{*}{60} & EL  &0.942   & 0.113 &0.935  & 0.194 &0.935&0.277   \\
                    & JEL   &0.943 & 0.108 &0.938 &0.189&  0.936&  0.260        \\
                    & Davidson &  0.940  & 0.218 & 0.918 &0.174& 0.935&  0.252     \\ \hline
\multirow{3}{*}{80} & EL   & 0.950&  0.168 & 0.951 &0.163&  0.956       &0.231     \\
                   & JEL   &0.948 & 0.172 &0.953 &0.164&  0.955&  0.225     \\
                    & Davidson &  0.940  & 0.172& 0.941 &0.152 & 0.945&  0.219    \\ \hline
\multirow{3}{*}{100} & EL    & 0.945& 0.161 & 0.949 & 0.165 &0.949&    0.198\\
                   & JEL   &0.945 & 0.160 &0.950 &0.151&  0.947&  0.205       \\
                    & Davidson &  0.941  & 0.160 & 0.933 &0.137 & 0.941&  0.197   \\ \hline

\end{tabular}}
\end{table}

The tables indicate that EL and JEL-based confidence intervals exhibit greater coverage probability, even with smaller sample sizes, in comparison to normal-based confidence intervals. The average width of these intervals is comparable; however, normal-based intervals are generally narrower. Incorporating of bias-corrected terms in the plug-in estimator influences the calculation of standard errors, leading to smaller estimated standard errors. And, further smaller estimated standard errors gives narrower confidence intervals in Davidson's method. However, narrower confidence intervals may be inaccurate and often lead to wrong conclusions if the coverage probability is insufficient. The coverage of the Davidson estimator for confidence intervals is comparatively lower making it unreliable in many instances. Thus, the proposed EL and JEL estimation methods for confidence intervals can address these problems in many applications.

\section{Illustration}

This section demonstrates the application of the proposed estimates for the Sen Index and the SST index on two distinct datasets. Initially, we conducted an analysis employing individual-level data from the Panel Study of Income Dynamics (PSID) for the survey years 2017, 2019, 2021, and 2023. The PSID is a longitudinal survey launched in 1968 by the Survey Research Center at the University of Michigan. This survey utilizes a genealogical design, selecting a representative sample of individuals whose descendants were surveyed over the years to examine the dynamics of poverty and economic well-being in the United States. Regular replenishment of the sample ensured its representativeness, rendering it appropriate for diverse academic studies.

This illustration uses individual data to aggregate various types of income, including labor market income, public and private transfer income, asset and capital income, as well as other income sources, for each individual in the dataset. This process results in a variable that encapsulates the total income received by each individual. Given that the individual level data is sourced from family level data, we have focused exclusively on those who receive income for the study's purpose.   The descriptive statistics of the income data for each year is considered in the Table \ref{tab:des_us}. Additionally, the probability distribution of income for each year is presented in Figure \ref{fig:den_us}.

\begin{table}
    \centering
     \caption{Descriptive statistics of individual incomes in US for different years}
    \begin{tabular}{ccccc}
    \hline
    \hline
   Year   &  2017& 2019 & 2021& 2023 \\
        \hline
    \hline

      $n$   &2374   & 2379 & 2007 & 2114\\
      \textit{Mean}  & 33736.970  & 34069.790 &39269.370&41479.280 \\

      \textit{SD} & 60326.610 & 48842.060 & 47397.250&46617.030 \\
     \textit{ Min }& 40   &6 &30&14\\
      \textit{Max} &2223000 & 1125000   &530000&900000  \\
      \textit{Range} & 2222960& 1124994 & 529970&899986\\
      \textit{Skewness} & 21.417& 8.650& 4.131&4.972 \\
     \textit{Kurtosis} & 734.750 &141.918 &27.535&60.898\\
      \hline
    \end{tabular}
    \label{tab:des_us}
\end{table}

The data presented in Table \ref{tab:des_us} and Figure \ref{fig:den_us} indicate a consistent increase in average income along with a decrease in the standard deviation of income over the years. This suggests a potential increase in economic well-being, while the decline in standard deviation may indicate an improvement in inequality within the population. These measures fail to adequately reflect the extent and severity of poverty over time. We calculated the Sen index and SST index using the income data for each year, and the estimates and their confidence intervals  are presented in Table \ref{tab:sst_sen_jel_us} and visually presented in Figure \ref{fig:ci_psid} . For the calculation for the Sen and SST index, we have used annual individual income to be USD 12000 as the poverty line. The U.S. Census Bureau establishes and annually updates the official poverty line thresholds. The following Table \ref{tab:opm_us} presents the Official Poverty Measure (OPM) for single individuals for each year selected for this study (U.S. Census Bureau, 2023). For the purpose of illustration, we selected the rounded OPM level of USD 12,000 for the year 2017 to facilitate comparison.

\begin{table}
    \centering
     \caption{OPM values of Single person incomes in US for different years (in USD)}
    \begin{tabular}{ccccc}

  \hline

         Years& 2017& 2019 & 2021&2023  \\
           \hline
          \hline

          \hline
          \hline
      OPM & 12,060&13,011&13,780&14,580\\

           \hline
          \hline

    \end{tabular}

    \label{tab:opm_us}
\end{table}

\begin{table}
\centering
\caption{SST and Sen  indices with jackknife empirical likelihood confidence intervals for different years}
\label{tab:sst_sen_jel_us}
\begin{tabular}{c ccc ccc}
\hline\hline
 & \multicolumn{3}{c}{\textbf{SST Index (U-statistic)}}
 & \multicolumn{3}{c}{\textbf{Sen Index (U-statistic)}} \\
\cline{2-4} \cline{5-7}
\textbf{Year}
& Estimate & JEL Lower & JEL Upper
& Estimate & JEL Lower & JEL Upper \\
\hline
2017 & 0.2692 & 0.2502 & 0.2892 & 0.1861 & 0.1727 & 0.2007 \\
2019 & 0.2584 & 0.2398 & 0.2780 & 0.1799 & 0.1667 & 0.1942 \\
2021 & 0.2382 & 0.2188 & 0.2588 & 0.1639 & 0.1503 & 0.1788 \\
2023& 0.2213 & 0.2028 & 0.2409 & 0.1520 & 0.1392 & 0.1661 \\
\hline
\hline
\end{tabular}
\end{table}

\begin{figure}
\centering
\caption{Graph showing the SST and Sen indices with jackknife empirical likelihood confidence intervals for different years}
\vspace{0.2in}
\includegraphics[width=7cm]{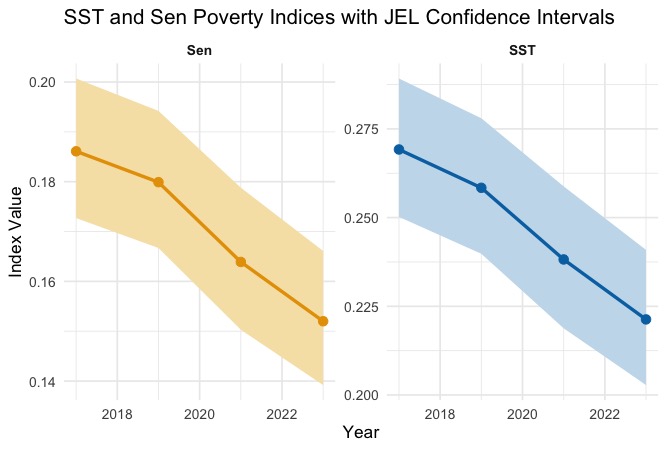}
\label{fig:ci_psid}
\end{figure}

The estimated Sen and SST index for each year indicates a consistent decline in poverty over time. Nonetheless, the outcome may appear counterintuitive for the year 2021. The wave of 2021 coincided with the pandemic period, during which evidence indicated a drop in economic activity and a decline in income globally. This data indicates a notable decrease in poverty in 2021 relative to 2019. A careful review of the literature reveals studies indicating an increase in income among the population during this period, attributed to various public financial transfers to individuals, such as unemployment benefits (Valletta et al., 2024).  To further validate this, we determined the proportion of unemployment benefits for each year, with the trend presented in Figure \ref{fig:unem}.

\begin{figure}
\centering
\caption{Proportion of unemployment benefit received by the  individuals incomes in US for different years}
\vspace{0.2in}
\includegraphics[width=7cm]{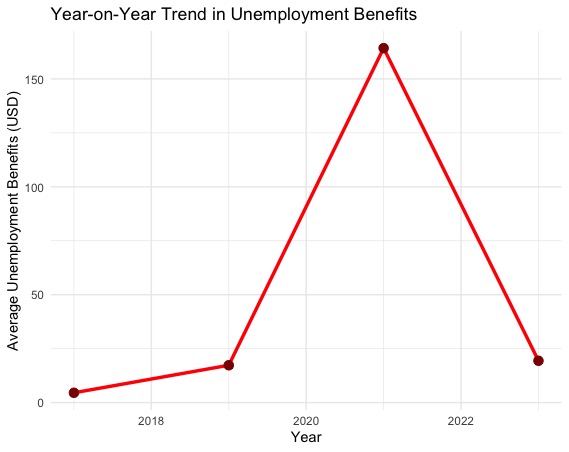}
\label{fig:unem}
\end{figure}

Consequently, it is evident that the proposed estimators successfully capture the patterns demonstrated by the data in a dependable way. The values of the proposed estimators are aligning with the existing Davidson estimates, providing additional support for the validation of the proposed estimator. It is important to highlight that the strength of the proposed estimator lies in its ability to provide better and more reliable inferences, particularly in scenarios where sample sizes are small.

Additionally, we are using a monthly household income data from India on a state-by-state basis for empirical illustration of the proposed estimators.    The household-level income data for each state is sourced from the Consumer Pyramids Household Survey (CPHS) conducted by the Centre for Monitoring Indian Economy (CMIE), and the dataset is accessible at \url{https://consumerpyramidsdx.cmie.com}. This is a systematic survey undertaken periodically to gather data on Indian household demographics, expenditures, assets, and perceptions. Annually, three data collection waves occur, each spanning four months. We utilized data from Wave 28, which includes information collected from January to April 2023 for the analysis. We consider Kerala, Tamil Nadu, and Bihar as the three states for the study due to their distinct levels of income disparity due to their underlying socio-economic conditions.

\begin{figure}
\centering
\caption{Income distribution of different states in India}
\vspace{0.2in}
\begin{tabular}{c c }
 \includegraphics[width=7cm]{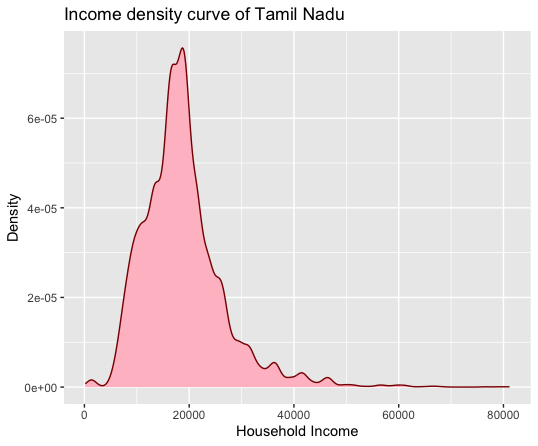}
    &\includegraphics[width=7cm]{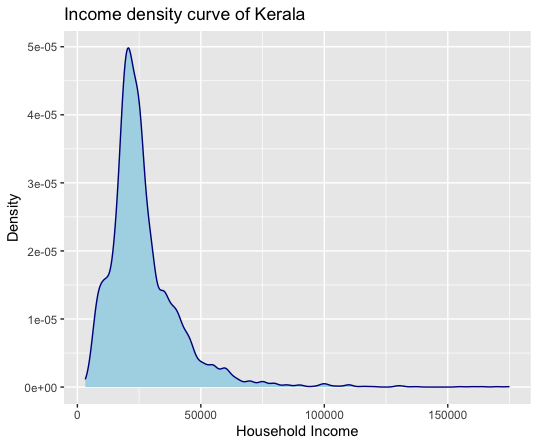}
\\\end{tabular}
\includegraphics[width=7cm]{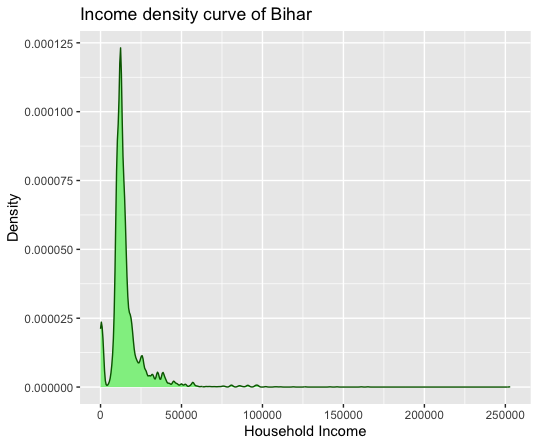}
\label{fig:den}
\end{figure}

\begin{table}
    \centering
     \caption{Descriptive statistics of Income from different states of India}
    \begin{tabular}{cccc}
    \hline
    \hline
     &   Kerala & Tamil Nadu & Bihar \\
        \hline
    \hline

      $n$   & 4310&   8129& 7475\\
      \textit{Mean}  & 26829.06 &  18736.31&15713.65\\

      \textit{SD} & 15185.06 &  7792.97 & 12018.63 \\
     \textit{ Min }&  3200& 165 &0\\
      \textit{Max} & 175000&  81150 &253000\\
      \textit{Range} & 171800&  80985& 253000\\
      \textit{Skewness} &2.75  &  1.53& 4.69 \\
     \textit{Kurtosis} & 14.4 &  5.19 &43.5\\
      \hline

    \end{tabular}

    \label{tab:des}
\end{table}

The income distribution for each state is shown in Figure \ref{fig:den}.  To enhance comprehension of the income distribution patterns among states, we have presented the descriptive statistics of the income data in Table \ref{tab:des}. The descriptive statistics indicate that the state of Kerala has a higher proportion of individuals in higher income brackets compared to Tamil Nadu and Bihar, with Bihar exhibiting a higher proportion of people in lower income strata. According to the report released by NITI Aayog (2023), a public policy think tank in India, utilizing data from the National Family and Health Survey (NFHS 5, conducted between 2019 and 2021), the hierarchy of poverty from highest to lowest is as follows: Bihar $>$ Tamil Nadu $>$ Kerala. Bihar ranks among the states with the highest poverty levels, while Kerala is among those with the lowest, making the selection of these states appropriate for this illustration. According to a study report by SBI Research (2025), the monthly poverty line is set at Rs 1,632 for individuals in rural regions and Rs 1,944 for individuals in urban areas for the fiscal year 2023-24, adjusted for inflation based on the official 2011-12 Tendulkar Committee poverty line. According to this concept, we are assuming a monthly poverty limit of Rs 2000 per person. Additionally, as the data pertains to household income, we assume the poverty line at Rs 8000 per month for a household, taking an average of four persons per household.   According to this assumption, the Sen and SST index and their confidence intervals are calculated using the proposed measures of Sen and SST are presented in the Table \ref{tab:sst_sen_jel_is} and visually presented in Figure \ref{fig:ci_hh}. Hence, we can see that the proposed measure can effectively capture the poverty in the states.
\begin{table}
\centering
\caption{SST and Sen  indices with jackknife empirical likelihood confidence intervals for different states of India}
\label{tab:sst_sen_jel_is}
\begin{tabular}{c ccc ccc}
\hline\hline
 & \multicolumn{3}{c}{\textbf{SST Index (U-statistic)}}
 & \multicolumn{3}{c}{\textbf{Sen Index (U-statistic)}} \\
\cline{2-4} \cline{5-7}
\textbf{State}
& Estimate & JEL Lower & JEL Upper
& Estimate & JEL Lower & JEL Upper \\
\hline
Kerala& 0.0050 & 0.0037 & 0.0058 & 0.0031 & 0.0025& 0.0040 \\
Tamil Nadu & 0.0150 & 0.0123& 0.0173 & 0.011 & 0.0093 & 0.0133 \\
Bihar & 0.1071 & 0.0981 & 0.1157 & 0.067 & 0.0616 & 0.0721 \\

\hline
\hline
\end{tabular}
\end{table}

\begin{figure}
\centering
\caption{Graph showing the SST and Sen indices with jackknife empirical likelihood confidence intervals for different states of India}
\vspace{0.2in}
\includegraphics[width=7cm]{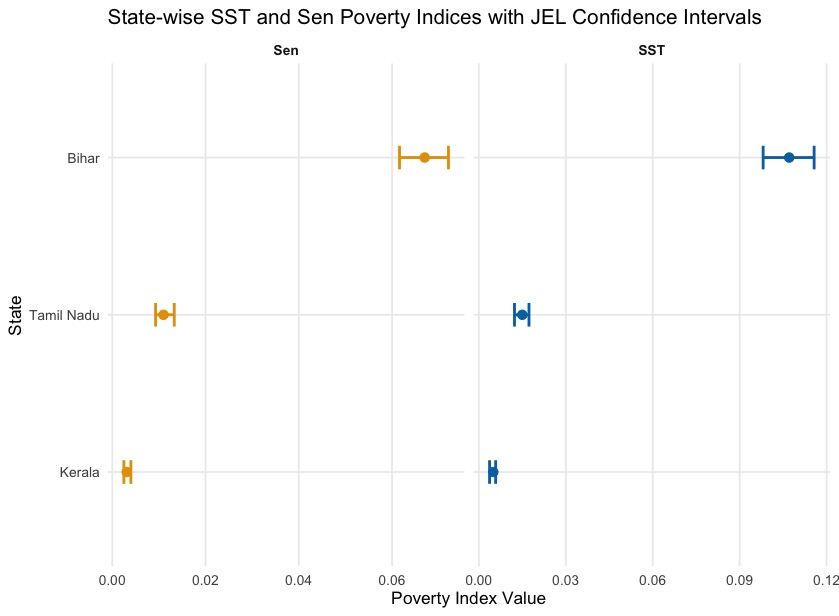}
\label{fig:ci_hh}
\end{figure}

\section{Conclusion}
This paper worked on a nonparametric inference framework for Sen and SST indices. Though, these indices play a pivotal role in poverty analysis, their non-linear structures ends up making the plug-in inference unreliable in finite sample scenario. Existing inference approaches rely primarily on large- sample approximations or bias correction techniques, which can lead to unreliable inferences when the income distributions are skewed.

We proposed new nonparametric estimators for both the Sen and SST indices using the
theory of $U$-statistic that have  a tractable variance structure under mild regularity conditions. The use of $U$-statistics theory, allowed us to establish asymptotic normality for the estimator, eliminating the need for a parametric assumptions on the income distribution. Further, relying on this, we developed the
jackknife empirical likelihood (JEL) inference for finding the confidence intervals for Sen and SST index.

Monte Carlo simulation indicated that proposed estimators exhibit substantially lower bias than Davidson (2009) estimator in small and moderate samples, while the JEL based confidence intervals achieve a appropriate coverage probability for a wide range of income distributions considered in the simulation study. The empirical illustration of the proposed estimators on the PSID data and Indian household income shows that the proposed approach can deliver stable inference there by helping in facilitating meaningful poverty comparisons across years and regions.

Thus, the paper contributes methodologically by providing a $U$-statistics based estimator which has tractable asymptotic properties. Additionally, it enhances inferential methods by proposing empirical likelihood-based procedures that are appropriate for nonlinear poverty measures in finite sample contexts frequently observed in various applications. The results can be extended to survey-weighted settings, dominance testing, and multidimensional poverty indices, thereby opening future research avenues.


 \section{Appendix}

 \noindent {\bf  Proof of Theorem 2:}
We know that
$$
\widehat{S}=\frac{2}{z}\frac{U_1}{U_2},
\qquad
S=\frac{2}{z}\frac{\Delta_1}{\Delta_2},
$$
where $U_1$ and $U_2$ are degree--2 $U$-statistics estimators  of $\Delta_1$ and $\Delta_2$ with symmetric kernels
$\psi_1$ and $\psi_2$, respectively.
By the unbiasedness property of U-statistics, we have
$$
\mathbb{E}(U_1)=\Delta_1,
\qquad
\mathbb{E}(U_2)=\Delta_2.$$
Let
$$
g(u,v)=\frac{u}{v}.
$$
Since $\Delta_2=F(z)=p>0$, a second-order Taylor expansion of
$g(U_1,U_2)$ around $(\Delta_1,\Delta_2)$ yields
\begin{align}
\frac{U_1}{U_2}
&= \frac{\Delta_1}{\Delta_2}
  + \frac{1}{\Delta_2}(U_1-\Delta_1)
  - \frac{\Delta_1}{\Delta_2^2}(U_2-\Delta_2)
  + R_n,
\label{taylor}
\end{align}
where the remainder term $R_n$ satisfies
$$
R_n
= O_p\!\left((U_1-\Delta_1)^2+(U_2-\Delta_2)^2
+|U_1-\Delta_1||U_2-\Delta_2|\right).
$$
Since $U_1$ and $U_2$ are U-statistics with finite second moments of their
kernels, it follows that
$$
U_1-\Delta_1 = O_p(n^{-1/2}),
\qquad
U_2-\Delta_2 = O_p(n^{-1/2}),
$$
and hence $R_n=O_p(n^{-1})$.

Taking expectations in \eqref{taylor} and using $\mathbb{E}(U_1-\Delta_1)=\mathbb{E}(U_2-\Delta_2)=0$, we obtain
$$
\mathbb{E}\!\left(\frac{U_1}{U_2}\right)=\frac{\Delta_1}{\Delta_2}+\mathbb{E}(R_n).
$$
Since
$\mathrm{Var}(U_1)=O(n^{-1})$ and $\mathrm{Var}(U_2)=O(n^{-1})$, the Cauchy--Schwarz inequality yields $$\mathbb{E}(R_n)=O(n^{-1}).$$
Therefore,
$$
\mathbb{E}\!\left(\frac{U_1}{U_2}\right)
\longrightarrow \frac{\Delta_1}{\Delta_2}.
$$
Multiplying both sides by $2/z$, we conclude that
$$
\mathbb{E}(\widehat{S}) \longrightarrow S,
$$
which proves that $\widehat{S}$ is an asymptotically unbiased estimator of the Sen index.\\

\noindent
{\bf Proof of Theorem 3.}
Recall that
$$
\widehat S=\frac{2}{z}\frac{U_1}{U_2},
\qquad
S=\frac{2}{z}\frac{\Delta_1}{\Delta_2},
$$
where $U_1$ and $U_2$ are degree--2 U-statistics estimating $\Delta_1$ and
$\Delta_2=F(z)>0$, respectively.

Consider,
$$
\sqrt n(\widehat S-S)
=
\frac{2\sqrt n}{z}
\frac{\Delta_2(U_1-\Delta_1)-\Delta_1(U_2-\Delta_2)}{U_2\Delta_2}.
$$
Since $U_2\xrightarrow{p}\Delta_2$, Slutsky’s theorem gives
$$
\sqrt n(\widehat S-S)
=
\frac{2}{z\Delta_2}
\Bigl[
\sqrt n(U_1-\Delta_1)
-\frac{\Delta_1}{\Delta_2}\sqrt n(U_2-\Delta_2)
\Bigr]
+o_p(1).
$$

By the joint central limit theorem for U-statistics,
$$
\sqrt n
\begin{pmatrix}
U_1-\Delta_1\\
U_2-\Delta_2
\end{pmatrix}
\xrightarrow{d}
N\!\left(
\begin{pmatrix}0\\0\end{pmatrix},
4
\begin{pmatrix}
\sigma_1^2 & \sigma_{12}\\
\sigma_{12} & \sigma_2^2
\end{pmatrix}
\right),
$$
where $\sigma_1^2$, $\sigma_2^2$, and $\sigma_{12}$ are determined by the
first-order Hoeffding projections.

For $U_1$, the projection is
$$h_1(X)
=
\mathbb E\{\psi_1(X,X_2)\mid X_1=X\}-\Delta_1
=
X(F(X)-F(z))-\int_0^X y\,dF(y).$$
Therefore, we have
$$\sigma_1^2=Var\left(
X(F(X)-F(z))-\int_0^X y\,dF(y)
\right).$$
For $U_2$, since $\psi_2(x_1,x_2)=\tfrac12\{I(x_1\le z)+I(x_2\le z)\}$,
$$
h_2(X)
=
\mathbb E\{\psi_2(X,X_2)\mid X_1=X\}-\Delta_2
=
\tfrac12\{I(X\le z)-F(z)\},
$$
which gives
$$\sigma_2^2
=
Var(h_2(X))
=
\frac14 F(z)\{1-F(z)\}.$$
The covariance term can be found using the first order projections.
Since
$$h_2(X)=\tfrac12\{I(X\le z)-F(z)\},$$we have
$$
\sigma_{12}
=
Cov(h_1(X),h_2(X))
=
\frac12\,Cov\!\bigl(h_1(X),\,I(X\le z)-F(z)\bigr).
$$
Because $\mathbb E\{h_1(X)\}=0$ and $F(z)$ is constant, this reduces to
$$
\sigma_{12}
=
\frac12\,Cov\!\left(
X(F(X)-F(z))-\int_0^X y\,dF(y),
\, I(X\le z)
\right).
$$
Denote
$$
A=X(F(X)-F(z))-\int_0^X y\,dF(y),
\qquad
B=I(X\le z).
$$
Using the definition of covariance, we have
$$
\sigma_{12}
=
\frac12\,Cov(A,B)
=
\frac12\Bigl[
\mathbb E\{A\,I(X\le z)\}
-
\mathbb E(A)\mathbb E\{I(X\le z)\}
\Bigr].
$$
Since $\mathbb E(A)=0$ and $\mathbb E\{I(X\le z)\}=F(z)$, this reduces to
$$
\sigma_{12}
=
\frac12\,\mathbb E\!\left[
\Bigl\{X(F(X)-F(z))-\int_0^X y\,dF(y)\Bigr\}
I(X\le z)
\right].
$$
Equivalently,
$$
\sigma_{12}
=
\frac12
\int_0^z
\left[
x\{F(x)-F(z)\}
-
\int_0^x y\,dF(y)
\right]
\,dF(x).
$$
Hence, we have the asymptotic variance, as stated in the theorem. \\

 \noindent {\bf  Proof of Theorem 4:} Define the centred estimating function
$$
g_i = K(X_i,S_0),
\qquad i=1,\dots,n.
$$
By construction, $E(g_i)=0$ and $E(g_i^2)<\infty$ by assumption.
Let$$
\bar{g}_n = \frac{1}{n}\sum_{i=1}^n g_i.
$$
A Taylor expansion of the Lagrange multiplier equation yields
\begin{equation}\label{lambda1}
  \lambda_1=\frac{\bar{g}_n}{E(g_i^2)} + o_p(n^{-1/2}).
\end{equation}Using this expansion,
\begin{eqnarray*}
    \log L(S_0)&=&-\sum_{i=1}^{n} \log\left(1+\lambda_1 g_i\right)\\
    &=&-\lambda_1 \sum_{i=1}^{n} g_i + \frac{1}{2}\lambda_1^2 \sum_{i=1}^n g_i^2 + o_p(1).
\end{eqnarray*}
Since $\sum_{i=1}^n g_i = n\bar{g}_n$, substituting the expression for  $\lambda_1$ in (\ref{lambda1}), from above equation, we obtain
$$
-2\log L(S_0)=\frac{n\bar{g}_n^2}{E(g_i^2)} + o_p(1).
$$
By the central limit theorem,
$$
\sqrt{n}\,\bar{g}_n
\;\xrightarrow{d}\;
N(0,\,E(g_i^2)).
$$
Hence, we have,
$$
\frac{n\bar{g}_n^2}{E(g_i^2)}
\,\xrightarrow{d}\,
\chi^2_1.
$$
Therefore $$-2\log L(S_0)
\;\xrightarrow{d}\;
\chi^2_1,$$
which completes the proof.

\noindent{\bf Proof of Theorem 5 :}
Using the central limit theorem of U-statistics (Hoeffding, 1948), as $n\rightarrow \infty$, we obtain
\begin{equation}\label{cltsenjel}
  \frac{\sqrt{n}\widehat{S}_n  }{2\sigma_{g}}\xrightarrow[]{d}N(0,1).
\end{equation}
\noindent The proof for the above result is similar to that of Theorem 3.

We first prove the elementary identity
$$
\frac{1}{n}\sum_{k=1}^n \widehat V_k = \widehat S_n,
$$
where
$$
\widehat S_n=\frac{1}{\binom{n}{2}}\sum_{1\le i<j\le n}\psi(X_i,X_j;\widehat S)
=\frac{2}{n(n-1)}\sum_{1\le i<j\le n}\psi_{ij},
$$
and the jackknife pseudo-values are defined by
$$
\widehat V_k = n\widehat S_n-(n-1)S_{n-1,k},\qquad k=1,\dots,n,
$$
with $S_{n-1,k}$ the U–statistic calculated  from the sample with the $k$-th observation deleted.
By the definition of the leave-one-out statistic,
\begin{eqnarray*}
    S_{n-1,k}
&=&\frac{1}{\binom{n-1}{2}}
\sum_{\substack{1\le i<j\le n\\ i,j\neq k}}\psi_{ij}\\
&=&\frac{2}{(n-1)(n-2)}
\sum_{\substack{1\le i<j\le n\\ i,j\neq k}}\psi_{ij}.
\end{eqnarray*}
Summing $S_{n-1,k}$ over $k=1,\dots,n$ gives
$$
\sum_{k=1}^n S_{n-1,k}
=\frac{2}{(n-1)(n-2)}
\sum_{k=1}^n
\sum_{\substack{1\le i<j\le n\\ i,j\neq k}}\psi_{ij}.
$$
For each unordered pair $(i,j)$ with $ 1\le i<j\le n$, the inner summand $\psi_{ij}$ is present in the inner sum for every $k$ except $k=i$ and $k=j$; hence it appears exactly $n-2$ times when summing over $k$. Therefore
$$
\sum_{k=1}^n
\sum_{\substack{1\le i<j\le n\\ i,j\neq k}}\psi_{ij}
=(n-2)\sum_{1\le i<j\le n}\psi_{ij},
$$
and consequently
$$\sum_{k=1}^n S_{n-1,k}
=\frac{2}{(n-1)(n-2)}(n-2)\sum_{1\le i<j\le n}\psi_{ij}
=\frac{2}{n-1}\sum_{1\le i<j\le n}\psi_{ij}.$$
Comparing with the definition of $\widehat S_n$ yields
$$\sum_{k=1}^n S_{n-1,k}=n\widehat S_n.$$
Now summing the  pseudo-values, we obtain
$$\sum_{k=1}^n \widehat V_k
=\sum_{k=1}^n\big(n\widehat S_n-(n-1)S_{n-1,k}\big)
= n^2\widehat S_n-(n-1)\sum_{k=1}^n S_{n-1,k}
= n^2\widehat S_n-(n-1)\,n\widehat S_n
= n\widehat S_n.$$
Dividing by $n$ gives the required identity.

Next, we present the standard asymptotic expansion of the JEL log-likelihood ratio.
Write $v_k=\widehat V_k$ and let
$$\overline v=\frac{1}{n}\sum_{k=1}^n v_k=\widehat S_n,\qquad
S_v^2=\frac{1}{n}\sum_{k=1}^n v_k^2.$$Denote by \(J(S)\) the JEL statistic (the empirical log-likelihood ratio formed from pseudo-values) and let \(\lambda_2\) be the Lagrange multiplier that solves
$$\frac{1}{n}\sum_{k=1}^n \frac{v_k}{1+\lambda_2 v_k}=0.$$
\noindent Assume
\begin{enumerate}
  \item $E\{\psi(X_1,X_2;\widehat S)^2\}<\infty$ and the usual regularity conditions for U–statistics hold;
  \item $\max_{1\le k\le n}|v_k|=o_p(\sqrt n)$ (cf.\ Lemma A.4, Jing et al.\ (2009)).
\end{enumerate}
Then
$$
\lambda_2=\frac{\overline v}{S_v^2}+o_p(n^{-1/2}),
$$
and
$$
J(S)=2n\lambda_2\overline v-nS_v^2\lambda_2^2+R_n
=\frac{n\overline v^{\,2}}{S_v^2}+o_p(1),
$$
where the remainder $R_n=o_p(1)$.
Expand the constraint equation around $\lambda_2=0$ gives us
$$
0
=\frac{1}{n}\sum_{k=1}^n\frac{v_k}{1+\lambda_2 v_k}
=\overline v - \lambda_2 \frac{1}{n}\sum_{k=1}^n v_k^2
+ \frac{1}{n}\sum_{k=1}^n v_k^3 \frac{\lambda_2^2}{1+\lambda_2 v_k}.
$$
Rearranging and substituting the expression for $S_v^2$, we obtain
$$
\lambda_2 \, S_v^2
= \overline v + \frac{1}{n}\sum_{k=1}^n v_k^3 \frac{\lambda_2^2}{1+\lambda_2 v_k}.
$$
By the boundedness condition $\max_k|v_k|=o_p(\sqrt n)$ and using the bound $S_v^2=O_p(1)$, the cubic remainder term is of order $o_p(n^{-1/2})$ (this follows from the fact that $n^{-1}\sum |v_k|^3 = o_p(\sqrt n)$, together with $|\lambda_2|=O_p(n^{-1/2})$. Hence
\begin{equation}\label{lambda2new}
    \lambda_2 = \frac{\overline v}{S_v^2} + o_p(n^{-1/2}).
\end{equation}
Next, expand the empirical log-likelihood ratio. That is, consider the Taylor expansion of $-2\sum\log(1+\lambda_2 v_k)$ given  by
$$J(S)
=2\sum_{k=1}^n\Big(\lambda_2 v_k - \tfrac{1}{2}\lambda_2^2 v_k^2 + \tfrac{1}{3}\lambda_2^3 v_k^3 \frac{1}{1+\theta_k\lambda_2 v_k}\Big)
=2n\lambda_2 \overline v - n S_v^2\lambda_2^2 + R_n,$$
where the remainder $R_n$ is $o_p(1)$ by the same bounds used above. Substituting the expansion for $\lambda_2$ given in (\ref{lambda2new}), above expression yields
$$J(S)=\frac{n\overline v^{\,2}}{S_v^2}+o_p(1).$$
Finally, by the central limit theorem for jackknife pseudo-values (or equivalently by the CLT for U–statistics and the delta method),
$$
\sqrt n\,\overline v \xrightarrow{d} N(0,\sigma_g^2),
$$
where
$\sigma_g^2=\lim_{n\to\infty}S_v^2$. Combining this with the expansion above and Slutsky's theorem yields
$$
J(S)\;=\;\frac{n\overline v^{\,2}}{S_v^2}+o_p(1)
\;\xrightarrow{d}\;
\chi^2_1,
$$
establishing the required limit for JEL ratio statistics.
\qed\\


\noindent {\bf  Proof Theorem 6:}
By the central limit theorem for U--statistics (see Lee, 1990), the statistic
$$
\widehat S_h=\frac{1}{\binom{n}{2}}\sum_{1\le i<j\le n} h(X_i,X_j)
$$
satisfies
$$
\sqrt{n}\,(\widehat S_h - S_h)\;\xrightarrow{d}\; N(0,4\sigma_2^2),
$$
where the asymptotic variance component $\sigma_2^2$ is the variance of the first-order projection
$$
\sigma_2^2=Var\left(\,E\big[h(X_1,X_2)\mid X_1\big]\,\right).
$$
In our case, the symmetric kernel is
$$h(X_1,X_2)
= I\!\left(\min(X_1,X_2)\le z\right)
- \frac{1}{z}\min(X_1,X_2)\,I\!\left(\min(X_1,X_2)\le z\right).$$
Define the indicator
$$
I_Z = I\!\left(\min(X_1,X_2)\le z\right).
$$
We compute the conditional expectation term-by-term.
For fixed $X_1=x$,
$$
E(I_Z\mid X_1=x)
=
P\!\left(\min(x,X_2)\le z\right)
=
\begin{cases}
1, & x\le z,\\[4pt]
F(z), & x>z.
\end{cases}
$$
Similarly,
$$
E\!\left(\min(X_1,X_2) I_Z \mid X_1=x\right)
=
\begin{cases}
x\,(1-F(x))+\displaystyle\int_0^{x} y\,dF(y), & x\le z,\\[10pt]
\displaystyle\int_0^{z} y\,dF(y), & x>z.
\end{cases}
$$
Combining the above two  expressions, we obtain
\begin{eqnarray*}
    E\!\left(h(X_1,X_2)\mid X_1=x\right)&=&
\left(F(z)-\frac{1}{z}\int_0^z y\,dF(y)\right) I(x>z)
\\
&&+\left(1-\frac{1}{z}\left(x(1-F(x))+\int_0^{x} y\,dF(y)\right)\right) I(x\le z).
\end{eqnarray*}Therefore
\begin{eqnarray*}
    \sigma_2^2
&=&Var\Bigg[
\left(F(z)-\frac{1}{z}\int_0^z y\,dF(y)\right) I(X>z)
\\&&+\left(1-\frac{1}{z}\left(X(1-F(X))+\int_0^{X} y\,dF(y)\right)\right) I(X\le z)
\Bigg].
\end{eqnarray*}
This completes the computation of the asymptotic variance $4\sigma_2^2$.


\begin{thebibliography}{222}
 \bibitem{} Bishop, J. A., Formby, J. P., and Zheng, B. (1997). Statistical inference and the Sen index of poverty. \emph{International Economic Review}, 38, 381--387.

\bibitem{} Clark, S., Hemming, R., and Ulph, D. (1981). On indices for the measurement of poverty. \emph{The Economic Journal}, 91(362), 515--526.

\bibitem{} Centre for Monitoring Indian Economy (CMIE). (2023). Consumer Pyramids Household Survey (CPHS). Mumbai: CMIE.

\bibitem{davidson2009reliable} Davidson, R. (2009). Reliable inference for the Gini index. \emph{Journal of Econometrics}, 150(1), 30--40.

\bibitem{} Foster, J., Greer, J., and Thorbecke, E. (1984). A class of decomposable poverty measures. \emph{Econometrica}, 52(3), 761--766.

\bibitem{} Hoeffding, W. (1948). A class of statistics with asymptotically normal distributions. \emph{The Annals of Mathematical Statistics}, 19, 293--325.

\bibitem{jing2009jackknife} Jing, B. Y., Yuan, J., and Zhou, W. (2009). Jackknife empirical likelihood. \emph{Journal of the American Statistical Association}, 104(487), 1224--1232.

\bibitem{} Kakwani, N. C. (1980). \emph{Income Inequality and Poverty}. New York: Oxford University Press (for the World Bank).

\bibitem{LAJ} Lee, A. J. (2019). \emph{U-Statistics: Theory and Practice}. New York: Marcel Dekker.

\bibitem{lehmann1951consistency} Lehmann, E. L. (1951). Consistency and unbiasedness of certain nonparametric tests. \emph{The Annals of Mathematical Statistics}, 22(2), 165--179.

\bibitem{} NITI Aayog. (2023). \emph{National Multidimensional Poverty Index: A Progress Review 2023}. Government of India.

\bibitem{owen1988empirical} Owen, A. B. (1988). Empirical likelihood ratio confidence intervals for a single functional. \emph{Biometrika}, 75(2), 237--249.

\bibitem{owen1990empirical} Owen, A. B. (1990). Empirical likelihood ratio confidence regions. \emph{The Annals of Statistics}, 18(1), 90--120.

\bibitem{} Panel Study of Income Dynamics. (2023). Panel Study of Income Dynamics (PSID): Main interview, 1968--2023 [Data set]. Institute for Social Research, University of Michigan. \url{https://psidonline.isr.umich.edu}

\bibitem{} Peng, L. (2011). Empirical likelihood methods for the Gini index. \emph{Australian \& New Zealand Journal of Statistics}, 53(2), 131--139.

\bibitem{} Ratnasingam, S., Wallace, S., Amani, I., and Romero, J. (2024). Nonparametric confidence intervals for generalized Lorenz curve using modified empirical likelihood. \emph{Computational Statistics}, 39(6), 3073--3090.

\bibitem{QY2013} Qin, G., Yang, B., and Hall, N. E. B. (2013). Empirical likelihood-based inferences for Lorenz curve. \emph{Annals of the Institute of Statistical Mathematics}, 65, 1--21.


\bibitem{} SBI Research. (2025). Consumption expenditure survey reveals a remarkable decline in rural poverty—Estimated rural poverty at 4.86\% and urban poverty at 4.09\% for FY24 [Press release]. \url{https://www.sbi.co.in/webfiles/uploads/files_2425/HCES_SBI%20Report_Jan25.pdf}

\bibitem{} Sen, A. (1976). Poverty: An ordinal approach to measurement. \emph{Econometrica}, 44, 219--231.



\bibitem{} Shorrocks, A. F. (1995). Revisiting the Sen poverty index. \emph{Econometrica}, 63(5), 1225--1230.

\bibitem{} Stigler, S. M. (1974). Linear functions of order statistics with smooth weight functions. \emph{The Annals of Statistics}, 2, 676--693.

\bibitem{tha} Takayama, N. (1979). Poverty, income inequality, and their measures: Professor Sen's axiomatic approach reconsidered. \emph{Econometrica}, 47(3), 747--759.

\bibitem{} Thomas, D. R., and Grunkemeier, G. L. (1975). Confidence interval estimation of survival probabilities for censored data. \emph{Journal of the American Statistical Association}, 70(352), 865--871.

\bibitem{} Thon, D. (1979). On measuring poverty. \emph{Review of Income and Wealth}, 25(4), 429--439.

\bibitem{} Thon, D. (1983). Lorenz curves and Lorenz coefficients: A sceptical note. \emph{Weltwirtschaftliches Archiv}, 119(2), 364--367.

\bibitem{} U.S. Census Bureau. (2023). \emph{Poverty Thresholds}.

\bibitem{} Valletta, R. G., and Yilma, M. (2024). Enhanced unemployment insurance benefits in the United States during COVID-19: Equity and efficiency. \emph{Federal Reserve Bank of San Francisco Working Paper 2024--15}.



\bibitem{} Watts, H. (1968). An economic definition of poverty. In D. P. Moynihan (Ed.), \emph{On Understanding Poverty} (pp. 316--329). New York: Basic Books.

\bibitem{xu2000} Xu, K. (1998). Statistical inference for the Sen–Shorrocks–Thon index of poverty intensity. \emph{Journal of Income Distribution}, 8(1), 143--152.



%
%
\end{thebibliography}

\end{document}